\documentclass[%
 reprint,
superscriptaddress,
 amsmath,amssymb,
 aps,
pra,
]{revtex4-2}

\usepackage{graphicx}
\usepackage{dcolumn}
\usepackage{bm}
\usepackage{amsfonts}
\usepackage{amsthm}
\usepackage{tikz}
\usepackage{physics} 
\usepackage{xcolor}
\usepackage{enumitem}

\graphicspath{{figures/}}

\usetikzlibrary{decorations.pathmorphing,arrows.meta,calc,positioning}
\definecolor{tensorpurple}{RGB}{110,40,180}
\definecolor{tensoryellow}{RGB}{238,190,40}
\definecolor{auxred}{RGB}{200,30,30}

\tikzset{
  leg/.style={line width=0.8pt},
  wleg/.style={leg,decorate,decoration={snake,amplitude=0.9pt,segment length=4.5pt}},
  idx/.style={font=\scriptsize},
  tensP/.style={draw,fill=tensorpurple,minimum size=8mm,rounded corners=0.8mm},
  tensY/.style={draw,fill=tensoryellow,minimum size=8mm,rounded corners=0.8mm},
  note/.style={font=\scriptsize,align=left},
  eqn/.style={font=\scriptsize}
}

\newtheorem{lemma}{Lemma}

\newcommand{\kk}[1]{\mathopen{}\lvert\mkern-2mu#1\rangle\!\rangle\mathclose{}}
\newcommand{\bb}[1]{\mathopen{}\langle\!\langle#1\mkern-2mu\rvert\mathclose{}}

\newif\ifhighlightrevisions
\highlightrevisionstrue
\newcommand{\rev}[1]{\ifhighlightrevisions{\color{black}#1}\else#1\fi}

\begin{document}

\title{Quantum Dimension Reduction of Hidden Markov Models}

\author{Rishi Sundar}
\affiliation{Department of Physics \& Astronomy, University of Manchester, Manchester M13 9PL, United Kingdom}
\affiliation{Centre for Quantum Science and Engineering, University of Manchester, Manchester M13 9PL, United Kingdom}

\author{Thomas J.~Elliott}
\affiliation{Department of Physics \& Astronomy, University of Manchester, Manchester M13 9PL, United Kingdom}
\affiliation{Department of Mathematics, University of Manchester, Manchester M13 9PL, United Kingdom}
\affiliation{Centre for Quantum Science and Engineering, University of Manchester, Manchester M13 9PL, United Kingdom}

\date{January 22, 2026}

\begin{abstract}
Hidden Markov models (HMMs) are ubiquitous in time-series modelling, with applications ranging from chemical reaction modelling to speech recognition. These HMMs are often large, with high-dimensional memories. A recently-proposed application of quantum technologies is to execute quantum analogues of HMMs. Such quantum HMMs (QHMMs) are strictly more expressive than their classical counterparts, enabling the construction of more parsimonious models of stochastic processes. However, state-of-the-art techniques for QHMM compression, based on tensor networks, are only applicable for a restricted subset of HMMs, where the transitions are deterministic. In this work we introduce a pipeline by which \emph{any} finite, ergodic HMM can be compressed in this manner, providing a route for effective quantum dimension reduction of general HMMs. We demonstrate the method on both a simple toy model, and on a speech-derived HMM trained from data, \rev{obtaining favourable memory--accuracy trade-offs in the examples studied, relative to a simple classical state-merging baseline}.
\end{abstract}

\maketitle

\section{Introduction}

Hidden Markov models (HMMs) provide a compact description of discrete-time
stochastic processes in terms of a finite internal memory that evolves
stochastically while emitting observable symbols
\cite{rabiner1986introduction,shalizi2001computational}. They are used widely across
scientific modelling and data analysis, including speech recognition and
other sequential data \cite{rabiner1989hmm, baldi1994hidden, ghahramani1996factorial, fine1998hierarchical, seymore1999learning,  krogh2001predicting, stanke2003gene, karlof2003hidden, bhar2004hidden, gammelmark2014hidden,  warden2018speech}. As learned HMMs
grow in size and connectivity, a central challenge is to reduce their
effective memory cost while retaining the observable statistics that make
them useful.

Classical reduction techniques, such as state merging, aim to simplify a model while preserving its predictions
\cite{stolcke1994bmm,stolcke1996merging,dupont2005links,dupont2000alergia,singh2004psr}.
In practice, however, these methods can involve large intermediate
representations or trade state reduction against loss of longer-range
predictive structure when applied to realistic learned models with many
states and rich transition patterns. This raises a basic question: given a
complex HMM that fits data well, how much memory is \emph{really} required to
reproduce its observable statistics, and can that memory cost be reduced in
a controlled way?

Quantum models of stochastic processes provide another route to memory
reduction. By encoding predictive information into non-orthogonal quantum
memory states, quantum simulators can reproduce the statistics of a classical
process while reducing the required memory~\cite{gu2012quantum, mahoney2016occam, riechers2016minimized, aghamohammadi2018extreme, binder2018practical, Liu_2019, Elliott_2021, adhikary2021quantum, elliott2022quantum, zonnios2025quantum, riechers2025identifiability}. The $q$-sample
construction makes this connection concrete by mapping the stationary
statistics of a process to a translationally invariant quantum state on an
infinite chain, with entanglement across a cut quantifying the memory
resources needed for sampling \cite{schuld2018quantum, Yang_2018, yang2024dimensionreductionquantumsampling}. For processes admitting a
finite-dimensional $q$-sample iMPS, standard uniform-MPS algorithms allow
controlled bond-dimension truncation, trading a small loss in fidelity rate
for a large reduction in bond dimension
\cite{vanderstraeten_tangent-space_2019,yang2024dimensionreductionquantumsampling}, corresponding to quantum models of stochastic processes with significantly reduced memory dimension compared to classical HMMs.

\rev{Born-machine approaches provide a useful neighbouring point of comparison, since both frameworks use tensor-network probability representations based on squared amplitudes. Our focus here is different: we start from a given finite-state HMM and study its compression into a lower-dimensional quantum sequential generator, rather than direct generative learning from data. This places the present work more naturally in the HQMM / \(q\)-sample setting of recurrent process models, while keeping a clear conceptual link to Born-machine methods. Related work on HQMMs has also considered the direct learning of sequential quantum models from data~\cite{srinivasan2017LearningHQMM}, which is again distinct from the present post-training compression setting.}

A key limitation is that existing iMPS compression results apply most cleanly
when the underlying process admits a \emph{deterministic} (or \emph{unifilar})
presentation, where the next internal state is fixed by the current state
and emitted symbol. In this case the $q$-sample admits a \emph{normal} iMPS representation with
a primitive transfer operator and a well-conditioned canonical form
\cite{Yang_2018,yang2024dimensionreductionquantumsampling,perezgarcia2007matrixproductstaterepresentations}.
HMMs learned directly from data, however, are typically non-deterministic, and naive
tensor-network constructions can yield non-normal iMPS (or mixed-state
tensor-network descriptions) for which stable infinite-system truncation is
not available \cite{yang2024dimensionreductionquantumsampling}. This creates
a practical gap between realistic learned HMMs and the class of truncated quantum models
accessible to current iMPS-based compression pipelines.

Here we bridge that gap. We introduce a dilation that augments the output
alphabet of any finite-state, stationary, ergodic HMM so that the resulting
model is deterministic while preserving all finite-length word statistics on
the original alphabet -- without increasing the model's memory dimension. The $q$-sample of the dilated process is
guaranteed to be representable as a normal iMPS, enabling stable variational truncation to a
target bond dimension $\tilde d$. From the compressed tensors we can reconstruct
an effective compressed quantum model of the original process.

We demonstrate the method on (i) a tunable non-deterministic
source family of toy models and (ii) a speech-derived HMM trained from data. In both cases we
obtain a controlled trade-off between bond dimension and distortion. The compressed
iMPS provides a direct specification of a quantum sampler with memory
dimension $\tilde d$. \rev{Our aim is thus to provide a post-training compression pipeline for general finite ergodic HMMs, rather than a direct data-driven training method for compressed quantum generators.}

\section{Background}

\subsection{Classical models of stochastic processes}

We consider a discrete-time stochastic process \(\mathcal P\) that generates a
bi-infinite sequence of random variables \(\{X_t\}_{t\in\mathbb Z}\), each
taking values in a finite alphabet \(\mathcal X\)~\cite{khintchine1934korrelationstheorie}. The process is fully
specified by the joint distribution
\(\Pr(\dots,X_{-1},X_{0},X_{1},\dots)\) over all times. In practice, modelling
and simulating such a process requires a finite description of the
statistical dependencies between past and future~\cite{crutchfield2012between}.

Hidden Markov models (HMMs) provide one such description~\cite{rabiner1986introduction, upper1997theory}. An HMM is specified
by a tuple \((\mathcal S,\mathcal X,\{T^{x}\})\), where \(\mathcal S\) is a
finite set of hidden states, \(\mathcal X\) is the set of observable
symbols, and the transition tensor \(T^{x}\) has elements
\begin{equation}
T^{x}_{s's} = \Pr(S_{t+1}=s',X_{t}=x \mid S_t=s)
\end{equation}
for all \(s,s'\in\mathcal S\) and \(x\in\mathcal X\). For each fixed state
\(s\) the matrices satisfy \(\sum_{s',x} T^{x}_{s's} = 1\), so that they
define a proper conditional distribution over next states and outputs.

A process has a \emph{deterministic} HMM representation if the next internal
state is uniquely determined by the current state and the emitted symbol~\cite{crutchfield1989inferring, shalizi2001computational}. In
terms of the transition structure this means that for every state \(s\) and
symbol \(x\) there is at most one successor state \(s'\) with nonzero
transition probability:
\begin{equation}
\bigl|\{\,s' : T^{x}_{s's} > 0 \,\}\bigr| \le 1
\quad \forall\, s\in\mathcal S,\; x\in\mathcal X.
\end{equation}
Equivalently, given the initial state and the observed symbol sequence, the present
internal state is uniquely determined. Deterministic models play a
central role in the branch of complexity science known as computational mechanics: the \(\varepsilon\)–machine corresponds to the minimal deterministic HMM that exactly reproduces a process,
and its memory cost is the Shannon entropy of the stationary distribution
over causal states, the statistical complexity \(C_{\mu}\)~\cite{crutchfield1989inferring, shalizi2001computational}, often interpreted as a measure of structure~\cite{crutchfield1994calculi, crutchfield1997statistical, crutchfield2012between}.
Non-deterministic HMMs, in contrast, can represent the same process with
fewer states, but at the cost of ambiguous internal state trajectories
conditioned on observations~\cite{lohr2009non, ruebeck2018prediction}.

Throughout this work we restrict attention to finite-state, finite-alphabet,
stationary and ergodic processes, so that a unique stationary distribution
over internal states exists and time-translation-invariant descriptions are
well defined.

\subsection{Quantum models of stochastic processes}

Quantum models extend this framework by encoding predictive information into
quantum memory states \(\{\ket{\sigma_j}\}\) associated with internal
configurations of the model~\cite{gu2012quantum, binder2018practical, Liu_2019}. The simulation
proceeds via a joint evolution of the memory system and an output register,
implemented by an isometry (equivalently, a unitary on a larger space).
Starting from an initial memory state, one applies this evolution and
measures the output register to obtain the next symbol \(x\), leaving the
memory updated to the post-measurement state conditioned on that outcome.
Repeated application of this procedure generates sequences with the same
statistics as the target process. Such a procedure has been experimentally-realised in multiple platforms~\cite{palsson2017experimentally, ghafari2019dimensional, ghafari2019interfering, wu2023implementing}.

More generally, a quantum hidden Markov model (QHMM) ~\cite{monras2010hidden, Monras16_Quantum, Elliott_2021, adhikary2021quantum, Fanizza24_Quantum, zonnios2025quantum, riechers2025identifiability} on an alphabet
\(\mathcal X\) can be specified by the tuple
\((\mathcal X,\mathcal H,\rho_0,\{\mathcal E_x\}_{x\in\mathcal X})\), where
\(\rho_0\) is a density operator on the memory Hilbert space \(\mathcal H\)
and \(\{\mathcal E_x\}\) is a quantum instrument: each \(\mathcal E_x\) is
completely positive and trace non-increasing, and
\(\sum_x \mathcal E_x\) is trace preserving. The probability of a word
\(w=x_1\cdots x_L\) is
\begin{equation}
\Pr(w)=\mathrm{Tr}\!\left(\mathcal E_{x_L}\circ\cdots\circ
\mathcal E_{x_1}(\rho_0)\right).
\end{equation}

The quantum statistical memory \(C_q\) is defined as the von Neumann
entropy of the stationary memory state,
\begin{equation}
C_q = S(\rho) = - \mathrm{Tr}(\rho \log_2 \rho), \quad
\rho = \sum_j \pi_j \ket{\sigma_j}\bra{\sigma_j},
\end{equation}
where \(\pi_j\) is the stationary probability of internal configuration
\(j\). Many processes admit models for which \(C_q \le C_{\mu}\), demonstrating a quantum
memory advantage, and for some families the separation between \(C_q\) and
\(C_\mu\) can be unbounded~\cite{garner2017provably, aghamohammadi2017extreme, elliott2018superior, elliott2019memory, Elliott_2020_extreme, elliott2021quantum, elliott2024embedding}.
Finite-dimensional quantum models can be constructed directly from the
classical HMM~\cite{Elliott_2021, riechers2025identifiability}.

\subsection{Matrix product states for deterministic models}

A stochastic process $\mathcal P$ over a finite alphabet $\mathcal X$ can be
encoded into a translationally invariant quantum state via its \emph{$q$-sample}
\cite{Yang_2018}. Informally, the $q$-sample is the infinite-chain limit of the
pure state whose computational-basis amplitudes are square roots of classical
word probabilities. For a length-$L$ block $x_{t:t+L}:=x_t x_{t+1}\cdots x_{t+L-1}$,
one considers
\begin{equation}
  \ket{P(X_{t:t+L})}
  := \sum_{x_{t:t+L}\in\mathcal X^L} \sqrt{P(x_{t:t+L})}\,\ket{x_{t:t+L}},
\end{equation}
and takes $L\to\infty$ in a manner compatible with stationarity.

When the $q$-sample admits an infinite matrix product state (iMPS)
representation with site tensors $\{A^x\}_{x\in\mathcal X}$ and transfer matrix
\begin{equation}
  E = \sum_{x\in\mathcal X} A^x \otimes (A^x)^*,
  \label{eq:transfer_background}
\end{equation}
all finite-block measurement statistics are obtained from a \emph{double-layer}
contraction. Concretely, if the leading eigenvalue of $E$ is nondegenerate,
then boundaries are irrelevant in the thermodynamic limit and one may express
the block distribution in terms of the leading left and right eigenmatrices
$V_l$ and $V_r$ of $E$ as \cite{Yang_2018}
\begin{equation}
  P_{\mathrm{MPS}}(x_{t:t+L})
  = \mathrm{Tr}\!\left(
  A^{x_{t+L-1}\dagger}\cdots A^{x_t\dagger}\, V_l\,
  A^{x_t}\cdots A^{x_{t+L-1}}\, V_r
  \right).
  \label{eq:pmps_double_layer}
\end{equation}
Equation \eqref{eq:pmps_double_layer} is the appropriate probability-extraction
formula for $q$-sample iMPS, and it replaces any single-layer expression of the
form $\langle l|A^{x_t}\cdots A^{x_{t+L-1}}|r\rangle$.

For  unifilar HMMs (in particular, $\varepsilon$-machines),
Yang \emph{et al.} show that the choice
\begin{equation}
  A^x_{s's} = \sqrt{T^x_{s's}}
  \label{eq:sqrt_site_matrices}
\end{equation}
fully represents the process, in the sense that
$P_{\mathrm{MPS}}(x_{t:t+L}) = P(x_{t:t+L})$ for all $L$, and moreover that the
process is ergodic if and only if the transfer matrix $E$ has a nondegenerate
leading eigenvalue \cite{Yang_2018}. In addition, the canonical form of an iMPS
can be constructed systematically from $V_l$ and $V_r$ via factorizations
$V_r=W_rW_r^\dagger$ and $V_l=W_l^\dagger W_l$ followed by an SVD of
$W_lW_r$ to obtain Schmidt coefficients and canonical tensors \cite{Yang_2018}.

In this work, we will use Eq.~\eqref{eq:pmps_double_layer} as the operative link
between iMPS tensors and classical block statistics, and we explicitly gauge-fix
our tensors into canonical form before interpreting them as Kraus operators.
A compatible canonical-form framework and completeness relations are reviewed in
Ref.~\cite{yang2024dimensionreductionquantumsampling} (Appendix A), where left-
and right-canonical tensors satisfy $\sum_x (A_l^x)^\dagger A_l^x = \mathbb I$
and $\sum_x A_r^x (A_r^x)^\dagger = \mathbb I$, enabling a channel (instrument)
interpretation of the site tensors \cite{yang2024dimensionreductionquantumsampling}.

\subsection{Difficulties with non-deterministic HMMs}

Direct application of MPS methods to non-deterministic HMMs encounters
significant challenges. A naive construction might attempt to use the same
element-wise square-root mapping \(T^{x} \mapsto A^{x}\) as above. For
non-deterministic models this typically yields a transfer operator \(E\) that
is not primitive: it may have multiple dominant eigenvalues or leading
eigenvectors that are not strictly positive. The resulting iMPS is non-normal,
so its canonical form is either not unique or poorly conditioned. Without
normality, standard truncation methods and tangent-space based infinite MPS
algorithms become inapplicable or numerically unstable, and straightforward
singular-value truncations along the virtual bonds often produce suboptimal
approximations~\cite{yang2024dimensionreductionquantumsampling}.

From the perspective of quantum implementations, one can represent general
non-deterministic processes using matrix product density operators (MPDOs) or
mixed-state tensor networks \cite{Elliott_2021, Banchi_2024, srinivasan2020quantumtensornetworksstochastic, srinivasan2017LearningHQMM}. These constructions are more flexible but less
suited to controlled dimension reduction, and they do not directly provide the
normal iMPS structure required by existing compression schemes. In the
remainder of this work we address these difficulties by constructing, for any
finite non-deterministic HMM, a dilated deterministic process whose
\(q\)–sample iMPS is normal by construction. This provides the starting point
for stable variational compression and for the compressed representations that
we analyse in the following sections.

\section{Theoretical Results}
\label{sec:method}

Our aim is to start from a general finite-state, discrete-alphabet, stationary
HMM and obtain a normal infinite matrix product state (iMPS) that can be
compressed by standard tensor-network methods. The overall pipeline has four
steps:
(i) dilate the original HMM to a deterministic process on an augmented output
alphabet;
(ii) construct its \(q\)–sample iMPS;
(iii) variationally truncate this iMPS to a reduced bond dimension; and
(iv) reconstruct an effective model for the original outputs and quantify the
approximation quality using the co-emission divergence rate (CDR).

Throughout we assume that the original HMM is finite-state, finite-alphabet,
stationary and ergodic.

\subsection{Dilation to a deterministic process}
\label{subsec:dilation}

Let \((\mathcal S,\mathcal X,\{T^{x}\})\) be an HMM for a stationary, ergodic
process \(\mathcal P\). In general this model need not be deterministic: for a
given state \(s\) and symbol \(x\), there may be several possible successor
states \(s'\) with \(T^{x}_{s's} > 0\).

We quantify this local branching by
\begin{equation}
k(s,x) = \bigl|\{\,s' \in \mathcal S : T^{x}_{s's} > 0\,\}\bigr|
\end{equation}
and define
\begin{equation}
d_y = \max_{s\in\mathcal S,\,x\in\mathcal X} k(s,x),
\end{equation}
so that \(1 \le d_y \le |\mathcal S|\). We introduce an auxiliary output
alphabet \(\mathcal Y = \{y_1,\dots,y_{d_y}\}\) and choose a labelling
function
\begin{equation}
f: \mathcal S \times \mathcal S \times \mathcal X \to \mathcal Y,
\end{equation}
which assigns a label \(y = f(s,s',x)\) to each allowed transition
\((s\to s')\) emitting symbol \(x\). The only constraint on \(f\) is that it
be injective for each fixed \((s,x)\):
\begin{equation}
s' \neq s'' \;\Longrightarrow\; f(s,s',x) \neq f(s,s'',x)
\quad \forall\, s\in\mathcal S,\; x\in\mathcal X.
\label{eq:f_injective}
\end{equation}
Such a labelling always exists, because \(d_y\) is at least as large as the
maximal branching \(k(s,x)\).

The dilated model has the same hidden-state space \(\mathcal S\) and a
composite output alphabet \(\mathcal X\times\mathcal Y\). Its transition
tensor is
\begin{equation}
T^{(x,y)}_{s's} =
\begin{cases}
T^{x}_{s's}, & \text{if } T^{x}_{s's} > 0 \text{ and } y = f(s,s',x),\\[4pt]
0, & \text{otherwise.}
\end{cases}
\label{eq:dilated_T}
\end{equation}

This dilation has three key structural properties:

\begin{itemize}
  \item \emph{Determinism:} for each state \(s\) and composite symbol
  \((x,y)\) there is at most one successor state \(s'\) with
  \(T^{(x,y)}_{s's}>0\).

  \item \emph{Preservation of observable statistics:} marginalising over the
  auxiliary outputs \(y\) reproduces the same finite-length block
  probabilities on \(\mathcal X\) as the original process \(\mathcal P\).

  \item \emph{Ergodicity:} if the original HMM is ergodic, then the dilated
  HMM is also ergodic, since the underlying Markov chain on \(\mathcal S\),
  obtained by summing over outputs, is unchanged.
\end{itemize}

Proofs are collected in Appendix~\ref{app:bounds}. The dilation therefore
yields a finite-state, stationary, ergodic and deterministic representation of
\(\mathcal P\) on the augmented alphabet \(\mathcal X\times\mathcal Y\).

\subsection{Constructing the \(q\)–sample iMPS}
\label{subsec:qsample}

Given the deterministic dilated HMM
\((\mathcal S,\mathcal X\times\mathcal Y,\{T^{(x,y)}\})\), we construct its
\(q\)–sample iMPS~\cite{Yang_2018} by taking the element-wise principal square
root of the transition tensor,
\begin{equation}
A^{(x,y)}_{s's} = \sqrt{T^{(x,y)}_{s's}} ,
\label{eq:A_from_T}
\end{equation}
for all \(s,s'\in\mathcal S\) and \((x,y)\in\mathcal X\times\mathcal Y\).
Treating the composite symbol \((x,y)\) as a single physical index, these
tensors define a translationally invariant iMPS on an infinite chain with
physical dimension \(d_{\mathrm{phys}} = |\mathcal X|\,|\mathcal Y|\).

The probability of a finite block
\((x_1,y_1),\dots,(x_L,y_L)\) is obtained by contracting the iMPS with
stationary boundary data, for example in the standard transfer-matrix form.
By construction and by the preservation property above, the marginal
distribution over \((x_1,\dots,x_L)\) coincides with that of the original
process \(\mathcal P\).

The corresponding transfer operator for the iMPS is
\begin{equation}
E = \sum_{x\in\mathcal X}\sum_{y\in\mathcal Y}
A^{(x,y)} \otimes \bigl(A^{(x,y)}\bigr)^{*},
\label{eq:transfer_dilated}
\end{equation}
where \({}^{*}\) denotes complex conjugation. Since the dilated HMM is
finite-state, stationary, ergodic and deterministic, the associated
\(q\)–sample iMPS has a \emph{primitive} transfer operator with a unique
leading eigenvalue \(\eta_0 = 1\) and strictly positive left and right
eigenvectors~\cite{Yang_2018}. The iMPS is therefore \emph{normal}, and
blocking a finite number of sites yields an injective MPS with a unique
canonical form and minimal bond dimension
\cite{perezgarcia2007matrixproductstaterepresentations}. In practical
computations we bring \(\{A^{(x,y)}\}\) into mixed canonical form using
standard gauge-fixing procedures.

The overall construction, from the original transition tensor through dilation
to the site tensors of the iMPS and the associated transfer operators, is
illustrated schematically in Fig.~\ref{fig:mps_pipeline}.

\begin{figure}[tb]
  \centering
    \includegraphics[]{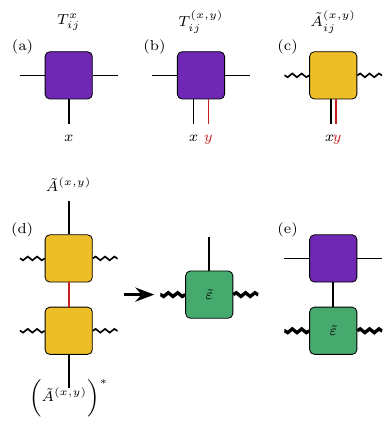}
  \caption{Tensor network representations of stochastic models and their compression. (a) A transition matrix $T^x_{ij}$ of a HMM can be represented as a rank-3 tensor; analogously, (b) the transition matrix of the dilated process $T^{(x,y)}_{ij}$ can be represented by a rank-4 tensor, as can (c) its element-wise square root $A^{(x,y)}_{ij}$. An  array of these form an injective iMPS when the two output legs are grouped, enabling the application of iMPS compression methods to deduce a dimension-reduced iMPS with site tensors $\tilde{A}_{ij}^{(x,y)}$. The (d) transfer matrix of the reduced iMPS defines a quantum channel $\tilde{\mathcal{E}}$ when a Kronecker-$\delta$ is applied to the two visible output legs, implementing a compressed model of the original process.  This can then (e) be paired with the original process to efficiently calculate their CDR.}
  \label{fig:mps_pipeline}
\end{figure}

\subsection{Variational truncation at reduced bond dimension}
\label{subsec:truncation}

The iMPS defined by \(\{A^{(x,y)}\}\) is an exact representation of the
dilated process, with bond dimension equal to the number of internal states
\(d_s = |\mathcal S|\) of the original HMM. To obtain a more compact model we
seek an approximate iMPS with a reduced bond dimension \(\tilde d < d_s\)
that is as close as possible, in the many-body sense, to the original.

We adopt a uniform tangent-space variational approach for iMPS with fixed bond
dimension~\cite{vanderstraeten_tangent-space_2019,yang2024dimensionreductionquantumsampling}.
Given the normal iMPS \(\{A^{(x,y)}\}\) with bond dimension \(d_s\), we
project it onto the manifold of translationally invariant iMPS with bond
dimension \(\tilde d\) by solving the associated optimisation problem for a
new set of tensors \(\{\tilde A^{(x,y)}\}\). In practice this is implemented
as an effective eigenvalue problem in the tangent space of the target
manifold, solved iteratively using standard MPS routines.

Each iteration involves contractions of the transfer operator \(E\) with
trial tensors of bond dimension \(\tilde d\). We need never form \(E\) explicitly
as a dense \(d_s^2\times d_s^2\) matrix; instead, we apply it as a linear map
with cost that scales polynomially in \(d_s\), \(\tilde d\) and the physical
dimension \(d_{\mathrm{phys}} = |\mathcal X|\,|\mathcal Y|\). \rev{The auxiliary
alphabet size satisfies \(|\mathcal Y| = d_y \le |\mathcal S|\), so in the worst
case the local physical dimension of the dilated iMPS can grow substantially
compared to the original alphabet size \(|\mathcal X|\). This enlarged local
dimension is the computational price paid for obtaining a deterministic
presentation whose \(q\)-sample admits a normal iMPS, and it can become the
dominant bottleneck in the variational truncation step for highly branching
models. A conservative dense upper bound for a single local transfer-operator
application is therefore of order \(\mathcal O(|\mathcal X|\,|\mathcal Y|\,d_s^2\tilde d^2)\),
with the precise prefactor depending on the chosen gauge and eigensolver
implementation. In practice, however, the effective cost is governed by the
number of allowed transitions inherited from the original
HMM, rather than by the full product \(|\mathcal X|\,|\mathcal Y|\); exploiting
this sparsity replaces the dense worst-case factor \(|\mathcal X|\,|\mathcal Y|\)
by a cost scaling linearly in the nonzero tensor contributions. For sparse
learned HMMs this can be substantially smaller than the worst-case scaling.} The outcome is
a compressed normal iMPS, specified by tensors \(\{\tilde A^{(x,y)}\}\), that
approximates the original dilated process while retaining translational
invariance. For algorithmic details of how the variational truncation approach for iMPS is performed, we refer the reader to Ref.~\cite{vanderstraeten_tangent-space_2019}, with its application to quantum compression of deterministic HMMs discussed in Ref.~\cite{yang2024dimensionreductionquantumsampling}

\subsection{Reconstructing an effective model}
\label{subsec:recons_model}

The truncated tensors $\{\tilde A^{(x,y)}\}$ define a translationally invariant
iMPS approximation to the dilated $(x,y)$-process. From this we can construct a compressed quantum model of the dilated process, from which we can in turn construct a dimension-reduced QHMM of the original process by coarse-graining over the dilation.

A crucial technical point is that Kraus (instrument) normalisation is a gauge
property of the iMPS representation. Following standard iMPS canonicalisation,
we first gauge-fix the truncated tensors into a canonical form obtained from
the leading left and right eigenmatrices of the truncated transfer operator
\cite{Yang_2018,yang2024dimensionreductionquantumsampling}. In left-canonical
form the site tensors satisfy the completeness relation
\begin{equation}
  \sum_{x\in\mathcal X}\sum_{y\in\mathcal Y}
  \bigl(\tilde A^{(x,y)}\bigr)^\dagger \tilde A^{(x,y)} = \mathbb I,
  \label{eq:left_canonical_completeness_xy}
\end{equation}
and hence define a completely positive trace-preserving (CPTP) map on the bond
space \cite{yang2024dimensionreductionquantumsampling}. This is the sense in
which the truncated iMPS specifies a physically valid sequential quantum
generator.

We now group the Kraus operators by the observed symbol $x$ by summing over the
auxiliary label $y$. For each $x\in\mathcal X$ define the completely positive
(CP) map
\begin{equation}
  \tilde{\mathcal E}_x(\rho)
  := \sum_{y\in\mathcal Y}
  \tilde A^{(x,y)}\,\rho\,\tilde A^{(x,y)\dagger}.
  \label{eq:cp_map_qhmm_rewrite}
\end{equation}
Equation \eqref{eq:left_canonical_completeness_xy} implies that the
unconditional channel $\tilde{\mathcal E}:=\sum_x \tilde{\mathcal E}_x$ is trace preserving.
Let $\tilde{\rho}_\star$ denote its stationary state, i.e., a fixed point satisfying
\begin{equation}
  \tilde{\rho}_\star = \sum_{x\in\mathcal X}\tilde{\mathcal E}_x(\tilde{\rho}_\star).
  \label{eq:rho_star_fixed_point_rewrite}
\end{equation}
Then $\{\tilde{\mathcal E}_x\}$ forms a quantum instrument that generates a stationary
process $\tilde{\mathcal P}$ on $\mathcal X$ via the usual update rule:
conditioned on observing $x_t$, the bond state updates as
\begin{equation}
  \tilde{\rho}_{t+1}
  = \frac{\tilde{\mathcal E}_{x_t}(\tilde{\rho}_t)}{\mathrm{Tr}\,\tilde{\mathcal E}_{x_t}(\tilde{\rho}_t)},
  \qquad
  \Pr_{\tilde{\mathcal P}}(x_t)=\mathrm{Tr}\,\tilde{\mathcal E}_{x_t}(\tilde{\rho}_t).
\end{equation}
Consequently, the probability of a word $w=x_1\cdots x_L$ is
\begin{equation}
  \Pr_{\tilde{\mathcal P}}(w)
  = \mathrm{Tr}\!\left(
  \tilde{\mathcal E}_{x_L}\circ\cdots\circ\tilde{\mathcal E}_{x_1}(\tilde{\rho}_\star)
  \right).
  \label{eq:word_prob_qhmm_rewrite}
\end{equation}

The dimension-reduced QHMM is then specified by the tuple $(\mathcal X,\tilde{\mathcal H},\tilde{\rho}_0,\{\tilde{\mathcal E}_x\}_{x\in\mathcal X})$, where $\tilde{\mathcal H}$ is the Hilbert space of the memory of the compressed QHMM, and $\tilde{\rho}_0$ is the initial state (generally, taken to be the steady-state $\tilde{\rho}_\star$).

\subsection{Quantifying distance from exact process}
\label{subsec:reconstruct_cdr}

For notational convenience we pass to the Liouville representation. Assign column-stacking
vectorisation $\mathrm{vec}$ defined on matrix units by
\begin{equation}
  \mathrm{vec}(\ket{i}\!\bra{j}) := \ket{i}\otimes\ket{j}.
\end{equation}
With this convention one has the standard identity
\begin{equation}
  \mathrm{vec}(A\rho B) = (A\otimes B^{\mathsf T})\,\mathrm{vec}(\rho),
  \label{eq:vec_convention}
\end{equation}
where ${}^{\mathsf T}$ denotes transpose in the computational basis.
Write $\kk{\rho}:=\mathrm{vec}(\rho)\in\mathbb C^{\tilde d^2}$.
Then the CP map $\tilde{\mathcal E}_x$ in Eq.~\eqref{eq:cp_map_qhmm_rewrite} is represented
by a linear operator $G^{(x)}$ via
\begin{equation}
  \kk{\rho'} = G^{(x)}\kk{\rho},
  \qquad
  G^{(x)} = \sum_{y\in\mathcal Y}
  \tilde A^{(x,y)} \otimes \bigl(\tilde A^{(x,y)}\bigr)^{*},
  \label{eq:Gx_liouville}
\end{equation}
where ${}^{*}$ denotes complex conjugation.
The trace functional is $\bb{\omega}:=\kk{\mathbb I}^{\dagger}$, so that
$\bb{\omega}\kk{\rho}=\mathrm{Tr}(\rho)$. Hence the triple
$(\mathcal X,\kk{\rho_\star},\{G^{(x)}\})$ defines a finite-state linear generator -- a generalised HMM (GHMM)~\cite{upper1997theory, riechers2025identifiability} -- where for any word $w=x_1\cdots x_L$,
\begin{equation}
  \Pr_{\smash{\tilde{\mathcal P}}}(w)
  = \bb{\omega}\, G^{(x_L)}\cdots G^{(x_1)} \kk{\rho_\star}.
  \label{eq:ghmm_form}
\end{equation}

The operators \(G^{(x)}\) need not be elementwise nonnegative or stochastic;
the requirement is simply that Eq.~\eqref{eq:ghmm_form} yields valid word
probabilities.

To quantify how well \(\tilde{\mathcal P}\) approximates the original process
\(\mathcal P\), we use the co-emission divergence rate (CDR)~\cite{Yang_2020}.
Let \(P\) and \(Q\) be stationary, ergodic processes described by finite-state
linear generators (either HMMs or GHMMs) \(\{L^x\}\) and \(\{\hat L^x\}\). We define the
associated transfer operators on the product state spaces by
\begin{align}
  E_P   &:= \sum_{x\in\mathcal X} L^x \otimes L^x, \nonumber\\
  E_Q   &:= \sum_{x\in\mathcal X} \hat L^x \otimes \hat L^x, \nonumber\\
  E_{PQ}&:= \sum_{x\in\mathcal X} L^x \otimes \hat L^x.
\end{align}

and let \(\mu_P\), \(\mu_Q\), and \(\mu_{PQ}\) denote their leading eigenvalues. The CDR is then given by
\begin{equation}
  R_C(P,Q)
  = -\frac{1}{2}\log_2 \frac{\mu_{PQ}}{\sqrt{\mu_P \mu_Q}}.
  \label{eq:cdr_def_ghmm}
\end{equation}
While calculated from GHMM representations, this expression depends only on the induced word distributions, valid for any stationary and
ergodic processes~\cite{Yang_2020}.

In our setting we take \(P\) to be the original HMM process \(\mathcal P\) on
\(\mathcal X\), with generators \(L^x = T^x\), and \(Q\) to be the compressed
process \(\tilde{\mathcal P}\) generated by \(\hat L^x = G^{(x)}\) as defined above. Their
CDR, \(R_C(\mathcal P,\tilde{\mathcal P})\), is the main observable figure of
merit we report in our numerics.

\subsection{Analytic bounds via fidelity divergence rates}
\label{subsec:bound}

The pipeline above produces, for each target bond dimension \(\tilde d\), a
compressed process \(\tilde{\mathcal P}\) together with its CDR
\(R_C(\mathcal P,\tilde{\mathcal P})\), which we evaluate numerically from
finite-state generators. Analytically, the dimension-reduction results of
Ref.~\cite{yang2024dimensionreductionquantumsampling} control
\emph{fidelity-type} errors of the underlying \(q\)--sample iMPS. We
therefore derive rigorous upper bounds on a \(\mathcal X\)-level
\emph{classical fidelity divergence rate} (CFDR), obtained from the dilated
\(q\)--samples by data processing. This CFDR is a different, but closely analogous,
distinguishability measure from the CDR \(R_C\) used in our plots.

Let \(\{\lambda_k\}_{k=1}^{d_s}\) denote the \emph{bond spectrum} of the
dilated \(q\)--sample iMPS across a single cut, i.e.\ the eigenvalues of the
stationary bond state \(\rho_\star\) in canonical form (equivalently, squared
Schmidt singular values), ordered so that
\(\lambda_1 \ge \lambda_2 \ge \cdots \ge \lambda_{d_s}\) and
\(\sum_k \lambda_k = 1\).
For a truncated bond dimension \(\tilde d\) define the discarded tail weight
\begin{equation}
  \varepsilon_{\tilde d}
  := \sum_{k > \tilde d} \lambda_k .
\end{equation}

We write \(R_F(\ket{P_{xy}},\ket{\tilde P_{xy}})\) for the quantum fidelity
divergence rate (QFDR) between the exact and truncated \((x,y)\)-dilated
\(q\)--sample states, and \(R_F(\mathcal P,\tilde{\mathcal P})\) for the
CFDR on \(\mathcal X\), defined by
\begin{equation}
  R_F(\mathcal P,\tilde{\mathcal P})
  := -\lim_{L\to\infty}\frac{1}{2L}\log_2
  \sum_{\vec x\in\mathcal X^L}\sqrt{P^{(L)}(\vec x)\,\tilde P^{(L)}(\vec x)} .
  \label{eq:cfdr_def}
\end{equation}
By applying dephasing in the \((x,y)\) basis followed by tracing out the
auxiliary register \(Y\), fidelity monotonicity implies the data-processing
bound
\begin{equation}
  R_F(\mathcal P,\tilde{\mathcal P})
  \;\le\;
  R_F(\ket{P_{xy}},\ket{\tilde P_{xy}}),
\end{equation}
proved in Appendix~\ref{app:bounds}. Ref.~\cite{yang2024dimensionreductionquantumsampling}
then implies that, for sufficiently small \(\varepsilon_{\tilde d}\), there
exists a constant \(c>0\) (depending on the local physical dimension and the
spectral gap of the transfer operator) such that
\begin{equation}
  R_F(\mathcal P,\tilde{\mathcal P})
  \;\le\; c\,\varepsilon_{\tilde d}.
  \label{eq:cfdr_tail_bound}
\end{equation}

The tail weight can in turn be bounded in terms of the bond entropy
\begin{equation}
  H(\lambda)
  := -\sum_k \lambda_k \log_2 \lambda_k ,
\end{equation}
giving (Appendix~\ref{app:bounds})
\begin{equation}
  \varepsilon_{\tilde d} \;\le\; \frac{H(\lambda)}{\log_2 \tilde d}
  \qquad (\tilde d \ge 2),
  \label{eq:tail_entropy_bound}
\end{equation}
and hence an entropy-based guarantee
\begin{equation}
  R_F(\mathcal P,\tilde{\mathcal P})
  \;\le\;
  c\,\frac{H(\lambda)}{\log_2 \tilde d}.
  \label{eq:cfdr_entropy_bound}
\end{equation}

Finally, we connect \(H(\lambda)\) to a simple algebraic quantity depending
on the dilation. Let \(K\) be the ``slice'' matrix formed by horizontally
concatenating all nonzero site tensors \(A^{(x,y)}\), viewed as linear maps
on the bond space,
\begin{equation}
  K = \bigl[\,A^{(x_1,y_1)}\;\; A^{(x_2,y_2)}\;\; \cdots \,\bigr] .
\end{equation}
The column space of \(K\) contains the support of the stationary bond state,
so \(\mathrm{rank}(K)\) upper-bounds the number of nonzero eigenvalues of
\(\rho_\star\), and therefore
\begin{equation}
  H(\lambda) \;\le\; \log_2 \operatorname{rank} K .
\end{equation}
Substituting into Eq.~\eqref{eq:cfdr_entropy_bound} yields the structural,
label-aware bound
\begin{equation}
  R_F(\mathcal P,\tilde{\mathcal P})
  \;\le\;
  c\,\frac{\log_2 \operatorname{rank} K}{\log_2 \tilde d}.
  \label{eq:cfdr_rank_bound}
\end{equation}

In the numerical results below we report the CDR \(R_C\), because it is
efficiently computable from finite-state generators\rev{: evaluating
Eq.~\eqref{eq:cdr_def_ghmm} requires only the leading eigenvalues of the
product-space transfer operators, which is a standard sparse-eigenvalue
problem}. The bounds above instead provide rigorous
certificates for the classical fidelity divergence rate \(R_F\)\rev{, a
Bhattacharyya-type overlap measure,} that is
directly controlled by iMPS truncation theorems. \rev{Both \(R_C\) and
\(R_F\) vanish if and only if the two processes are identical, but they are
distinct measures of process discrepancy and are not in general proportional.}
Without additional assumptions one should not interpret these as direct upper bounds on \(R_C\).
\rev{Thus, the analytic results provide rigorous worst-case control for the
CFDR, while the CDR serves as the main empirical diagnostic in our numerics.}

\section{Examples and Analysis}
\label{sec:results}

We now demonstrate and analyse the dilation--compression methodology on two examples.
First, we consider a \emph{Tunable Nondeterministic Source} (TNS), a generalisation of a non-deterministic generator~\cite{marzen2015informational} that has previously been used to exhibit quantum
reductions in memory \cite{Elliott_2021}. Here we use it as a controlled test bed
for how compression performance depends on internal parameters and on the
entanglement structure induced by the dilation. Second, we apply the method to a
hidden Markov model trained on a real speech dataset and compare the compressed
representations to a standard classical reduction baseline.

\subsection{Illustrative model: $N$-state simple non-deterministic source}
\label{subsec:sns}

As a controlled test bed we use a generalised $N$-state TNS with internal states
$\mathcal S=\{0,1,\dots,N-1\}$ and a single parameter
$p\in(0,1)$ tuning the transition structure. For generic $p$ and $N\geq2$ the model
is non-deterministic: conditioning on the emitted symbol does not uniquely fix
the successor state. The transition structure is shown schematically in
Fig.~\ref{fig:nstate-sns}.

\begin{figure}[t]
  \centering
   \includegraphics[width=\columnwidth,height=0.85\textheight,keepaspectratio]{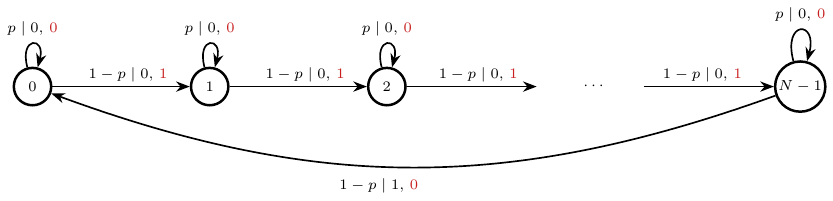}

  \caption{HMM representation of the generalised $N$-state tunable non-deterministic source (TNS).
  Transition labels show `probability $\mid$ emitted symbol', with the auxiliary label symbol in red.}
  \label{fig:nstate-sns}
\end{figure}

Starting from this TNS HMM we construct the dilated deterministic model according
to Sec.~\ref{subsec:dilation}, build its $q$--sample iMPS as in
Sec.~\ref{subsec:qsample}, and perform variational truncation at a sequence of
reduced bond dimensions $\tilde d$ following Sec.~\ref{subsec:truncation}. For
each choice of $(N,p)$ and target $\tilde d$ we reconstruct the effective model
for the original alphabet and compute the co-emission divergence rate
$R_C(\mathcal P,\tilde{\mathcal P})$ between the original and compressed
processes, as described in Sec.~\ref{subsec:reconstruct_cdr}.

Figure~\ref{fig:sns_cdr_vs_d} shows the resulting CDR as a function of the
truncated bond dimension $\tilde d$ for several representative parameter
choices. Each curve corresponds to a fixed pair $(N,p)$. In all cases the CDR
decreases systematically as $\tilde d$ increases, indicating that the
variational procedure produces families of compressed models that converge
toward the original process as more virtual resources are allowed. The rate of
decrease depends strongly on both $N$ and $p$: some parameter regimes reach a
given CDR at much smaller $\tilde d$ than others, reflecting differences in
how strongly the process resists low-dimensional compression. This behaviour is
also in line with the analytic picture of Sec.~\ref{subsec:bound}, where the
loss induced by truncation is controlled by the tail of the bond spectrum of
the (dilated) iMPS.

\begin{figure}[tb]
    \centering
    \includegraphics[width=\columnwidth]{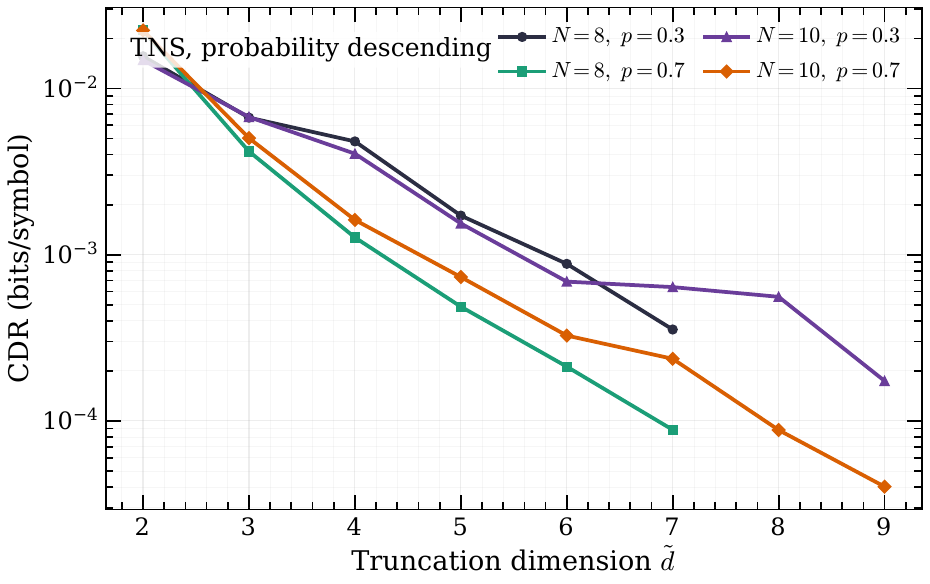}
    \caption{TNS model compression results: co-emission divergence rate $R_C$ between the original process and the
    reconstructed effective model as a function of truncated bond dimension $\tilde d$.
    Each curve corresponds to a different choice of the number of states $N$ and internal parameter $p$.}
    \label{fig:sns_cdr_vs_d}
\end{figure}

\begin{figure}[tb]
    \centering
    \includegraphics[width=\columnwidth]{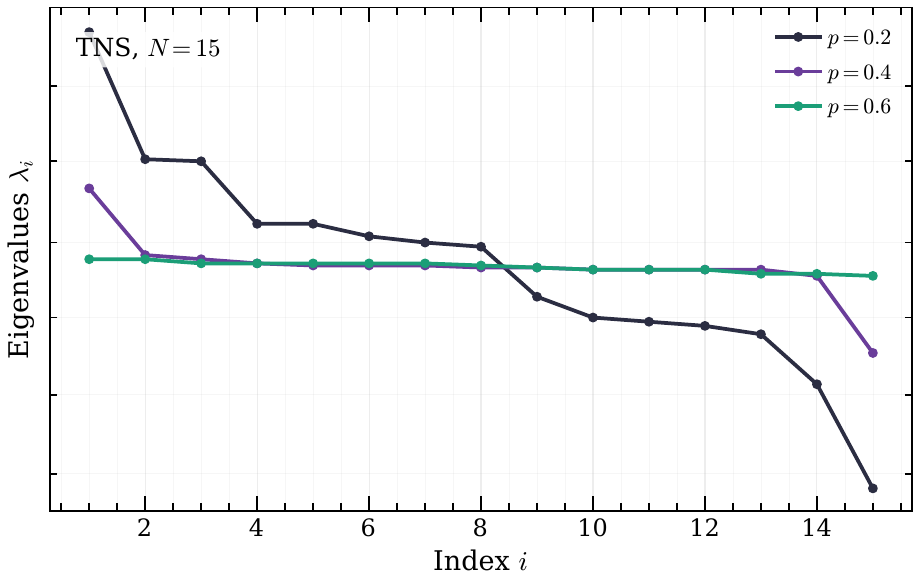}
    \caption{Schmidt spectra of the dilated iMPS of the TNS for fixed $N = 15$ and varied $p$, ordered by decreasing magnitude. The shape of the specra can vary, but parameter regimes with faster-decaying spectra are expected to be more amenable to truncation to small bond dimension.}
    \label{fig:sns_schmidt}
\end{figure}
To understand this dependence more structurally, we examine the Schmidt spectra
of the dilated iMPS before truncation. For each $(N,p)$ we bring the iMPS into
canonical form and compute the singular values across a single bond, which
quantify the bipartite entanglement across that cut. Figure~\ref{fig:sns_schmidt}
shows the Schmidt coefficients as a function of their index for a fixed $N$ and varied $p$. The decay patterns vary markedly: some spectra are
rapidly decaying (concentrating most weight into a few leading modes), while
others are much flatter. This also motivates the use of the variational method, as standard SVD truncation methods which work by discarding singular values work poorly on flat distributions.

Comparing Figs.~\ref{fig:sns_cdr_vs_d} and~\ref{fig:sns_schmidt}, one finds that
parameter choices with more rapidly decaying Schmidt spectra tend to achieve low
CDR at smaller $\tilde d$, while flatter spectra require larger $\tilde d$ to
reach comparable accuracy. This is consistent with the intuition that the
entanglement encoded in the dilated iMPS is a practical proxy for the
structural complexity that survives the dilation and resists compression.

\subsection{Application to a speech-derived HMM}
\label{subsec:speech}

To assess the practical utility of the method we apply it to a hidden Markov
model trained on a subset of the Google Speech Commands dataset
\cite{warden2018speech}. Each audio clip is mapped to a sequence of acoustic
feature vectors using Mel-frequency cepstral coefficients (MFCCs), including
time-derivative features \cite{davis1980mfcc}. We then discretise these
features by vector quantisation using MiniBatch $k$-means, yielding a finite
output alphabet \cite{sculley2010webscale}. Finally, we fit a categorical
(edge-emitting) HMM by maximum likelihood via the expectation-maximisation (or \emph{Baum--Welch}) procedure
\cite{dempster1977em,rabiner1989hmm}. The resulting learned model is strongly
non-deterministic, making it a realistic test case for the dilation step.

\begin{figure}[tb]
    \centering
    \includegraphics[width=\columnwidth]{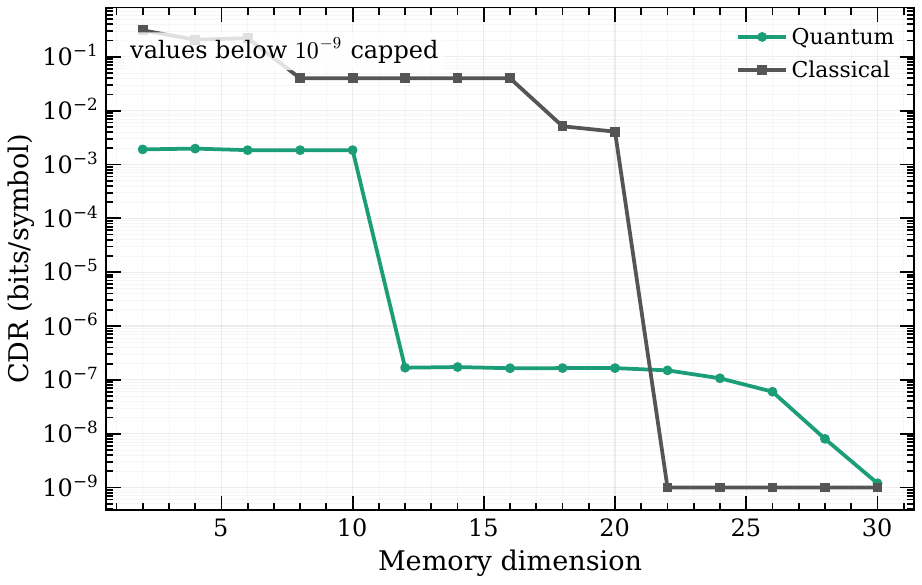}
    \caption{Speech-derived HMM compression results: comparison of quantum dimension reduction and a classical greedy state-merging baseline. Horizontal axis reports memory dimension as classical state count or quantum bond dimension. Vertical axis reports CDR $R_C$ (bits per symbol,
    log scale). Values below $10^{-9}$ are displayed at $10^{-9}$ for visibility.}
    \label{fig:speech_quantum_vs_classical}
\end{figure}

We take this trained HMM as the target process $\mathcal P$ and apply our dilation-compression pipeline to obtain a compressed quantum model of the process. To benchmark against a classical reduction pipeline, we also perform greedy
state merging directly on the original learned HMM. In each merge step we
choose the pair of states that minimises a stationary-weighted KL objective
defined on the one-step predictive distributions
$p_s(x,s') = T^{x}_{s's}$, and we update the merged state by a simple
stationary-mixture lumping rule. We then evaluate the CDR between the
original HMM and each merged HMM at the corresponding state counts.

We emphasise that this merging routine is intended as a transparent and
computationally inexpensive \emph{baseline} rather than a best-in-class classical
reduction method. It is greedy, uses a local one-step objective, and does not
optimise sequence-level divergences such as the CDR directly; more sophisticated
classical reductions (e.g.\ global objectives, multi-step lookahead, or spectral
methods) may achieve stronger performance in some regimes. The purpose of the
baseline is to provide a concrete reference curve against which the effect of
the dilation--compression pipeline can be assessed on the same learned model.

Figure~\ref{fig:speech_quantum_vs_classical} compares the quantum reduction curve
(memory dimension $\tilde d$) with the classical merging curve (retained HMM state count).
In the strongly compressed regime (memory dimension $\lesssim 20$), the quantum models
achieve substantially smaller CDR, often by orders of magnitude, than this greedy classical
baseline at the same nominal memory dimension. Notably, the quantum curve exhibits a sharp
improvement around $\tilde d\approx 12$, consistent with the variational truncation needing only few states to capture the dominant
modes of the dilated iMPS. By contrast, the classical baseline only approaches the numerical
floor once the retained state count is sufficiently large (here around $\gtrsim 22$ states),
suggesting that a moderate number of learned states can be merged with negligible sequence-level
distortion, whereas more aggressive classical reductions rapidly degrade performance.
\rev{We note that this comparison is against one specific classical baseline on the same learned model; more sophisticated classical reduction methods may narrow the gap in some regimes.}

\subsection{Impact of the labelling function}

\label{subsec:labelling}

The dilation procedure of Sec.~\ref{subsec:dilation} contains a genuine degree
of freedom: the choice of the labelling function
$f:\mathcal S\times\mathcal S\times\mathcal X\to\mathcal Y$ used to assign
auxiliary outputs to transitions. Different admissible labellings (all satisfying
the injectivity constraint) can induce different entanglement structures in the
dilated iMPS. Since truncation performance is sensitive to the bond spectrum, it
is natural to ask how strongly the choice of $f$ affects compressibility in
practice.

To probe this, we fix a TNS instance and compare several simple labelling
strategies that differ only in how they assign the auxiliary symbols $y$ to the
allowed transitions, while keeping the auxiliary alphabet size $|\mathcal Y|$
fixed. For each choice of $f$ we construct the corresponding dilated iMPS,
perform variational truncation across a range of bond dimensions $\tilde d$,
reconstruct the effective models, and compute the CDR.

Figure~\ref{fig:labelling_cdr}(a) shows $R_C(\mathcal P,\tilde{\mathcal P})$ as a
function of $\tilde d$ for the different labelling strategies. The curves differ
markedly: some labellings yield low CDR already at small $\tilde d$, while
others require substantially larger bond dimensions to reach comparable
accuracy. In this sense the labelling is not merely cosmetic; it acts as an
optimisation parameter that can materially change the achievable
CDR--$\tilde d$ trade-off. This observation also matches the structural picture
in Sec.~\ref{subsec:bound}, where label-dependent quantities (such as the slice
structure) can upper-bound entropic tails and hence influence error guarantees.

\begin{figure}[tb]
  \centering
    \includegraphics[width=\columnwidth]{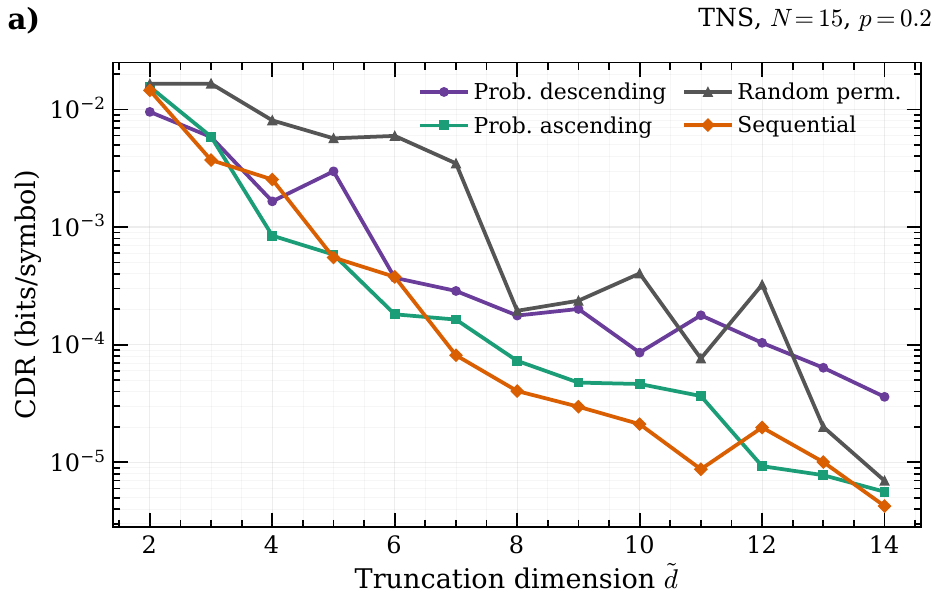}

\includegraphics[width=\columnwidth]{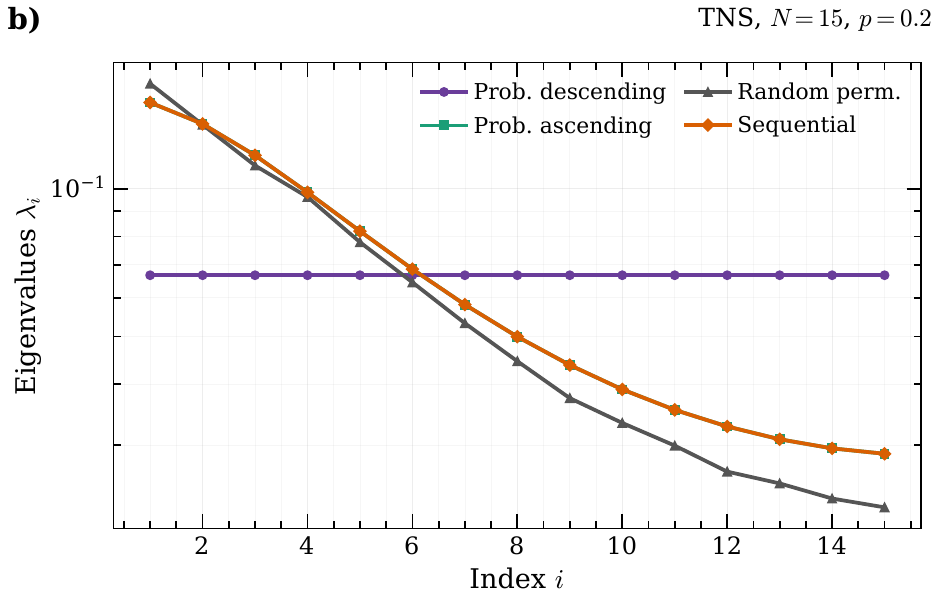}

  \caption{(a) Impact of labelling function on CDR for TNS model  with fixed parameters $(N,p)$. Each curve corresponds to a different strategy for assigning auxiliary labels $y$ to transitions. (b) Schmidt spectra of the dilated iMPS for each of the above labelling strategies; different assignments of auxiliary labels $y$ produce distinct  entanglement spectra that correlate with compression performance.
}
  \label{fig:labelling_cdr}
\end{figure}

The corresponding Schmidt spectra of the dilated iMPS for the same labellings
are shown in Fig.~\ref{fig:labelling_cdr}(b). Different assignments of $y$
produce distinct decay patterns. Labellings with more rapidly decaying spectra
tend to yield the lowest CDR at small $\tilde d$, while flatter spectra yield
poorer compression performance, reinforcing the role of the induced entanglement
structure as a practical predictor of compressibility.
\rev{Structurally, different admissible labellings all preserve the original HMM
exactly before truncation, but they redistribute the nonzero transition
probabilities across the $(x,y)$ slices of the dilated tensor. This
reorganisation changes the column structure of the slice matrix $K$ and hence
the Schmidt spectrum of the dilated iMPS, which in turn controls the
bond-dimension-versus-error trade-off after truncation. The specific labelling
heuristics used in our numerical study are documented in
Appendix A.}

A full optimisation of $f$ over all admissible assignments is a nontrivial
combinatorial problem and is beyond the scope of the present work. The examples
here show, however, that even simple heuristic labellings can significantly
change the practical trade-off between bond dimension and CDR, suggesting that
learning or optimising good labelling strategies is an important direction for
future research.

\section{Discussion and Conclusion}
\label{sec:discussion}

We introduced a general framework for mapping non-deterministic HMMs to matrix product state representations, enabling techniques for their implementation with compressed quantum models to be applied. We achieve with a dilation procedure that takes any finite, stationary, ergodic
HMM to a deterministic model on an enlarged output alphabet that preserves all
observable statistics on the original symbols. The $q$--sample of this dilated
process is representable by a normal iMPS with bond dimension $d_s = |\mathcal S|$, which can be
compressed using standard uniform-MPS tangent-space methods. From the compressed
tensors we reconstruct quantum model of the original process, with reduced memory dimension.

We applied this pipeline to a TNS toy model and to a
speech-derived HMM learned from data. In both cases we obtain families of
compressed models that interpolate smoothly between the original process and
low-bond approximations. Empirically the CDR decreases monotonically with
$\tilde d$, and there is a broad regime in which $\tilde d \ll d_s$ while
$R_C(\mathcal P,\tilde{\mathcal P})$ remains small. This supports the practical
value of the dilation--compression pipeline as a dimension-reduction tool for
practical quantum compression of HMMs

The examples also clarify how compressibility depends on the entanglement
structure of the dilated $q$--sample. For the TNS family, different choices of
the internal parameter $p$ produce markedly different Schmidt spectra: rapidly
decaying spectra admit accurate truncations at small $\tilde d$, while flatter
spectra require larger bond dimensions to achieve comparable CDR. This supports
the view that entanglement in the dilated iMPS is a useful proxy for structural
complexity that resists compression.

\rev{We emphasise that the present work establishes a general compression
route for arbitrary finite ergodic HMMs and demonstrates it on representative
examples, rather than providing a comprehensive large-scale benchmarking study.
The computational cost of the pipeline is dominated by the variational
truncation step, whose per-iteration cost scales with the bond dimensions
$d_s$ and $\tilde d$ and the effective physical dimension of the dilated iMPS.
A conservative dense estimate for the local contractions entering one
transfer-operator application is \(\mathcal O(|\mathcal X|\,|\mathcal Y|\,d_s^2\tilde d^2)\),
so because the auxiliary dilation can enlarge the physical dimension up to
$|\mathcal X|\,|\mathcal Y|$ with $|\mathcal Y| \le |\mathcal S|$, this cost
can grow substantially for highly branching models with many states. In
practice, the sparsity of the original transition tensor provides a
significant mitigation, since only nonzero tensor contributions enter the
contraction cost, effectively replacing the worst-case factor
$|\mathcal X|\,|\mathcal Y|$ by the number of allowed transitions.
Nevertheless, scaling the pipeline to HMMs with hundreds
or thousands of states will require careful exploitation of sparsity and
potentially approximate contraction schemes, and we regard this as an
important direction for future engineering of the method.}

The most immediate extensions of this work are practical. First, the labelling
function $f$ is a genuine design degree of freedom: our results show that it
can reshape the entanglement structure of the dilated $q$--sample and
substantially change the achievable CDR at fixed $\tilde d$. This suggests an
algorithmic direction in which the dilation is not treated as a fixed
preprocessing step, but as something to be chosen (or learned) to produce
dilations that are intrinsically easier to compress.

Second, the dilation--compression pipeline provides a systematic route from
large learned HMMs to small quantum samplers with an explicit and interpretable
resource parameter $\tilde d$. This opens the door to using tensor-network
diagnostics (bond spectra, entropies, canonical form data) as practical
predictors of which learned models are likely to admit aggressive compression
without significantly distorting sequence statistics.

\rev{It is also useful to situate this work relative to tensor-network
Born-machine approaches~\cite{han2018unsupervised,Stoudenmire2017,glasser2019expressive}, which also use squared
amplitudes of matrix product states to define probability distributions.
Born-machine methods are typically framed as direct generative learning from
data, without requiring an intermediate classical model. By contrast, our
pipeline takes a given HMM as input and compresses it into a lower-dimensional
quantum sequential generator with explicit structural control over the
approximation error. These are parallel approaches with different modelling
emphases. Training compressed quantum sequential models directly from raw
data, bypassing the intermediate HMM, is an interesting open direction that
is more naturally aligned with Born-machine methodology than with the present
post-training compression setting. It would be interesting in future work to
explore whether the compressed generators produced by our pipeline could also
serve as informed initialisations for subsequent data-driven optimisation.}

\rev{Finally, while we reported CDR because it is operationally meaningful and efficiently computable from finite-state generators, the analytic truncation guarantees established here apply to a fidelity-type divergence rate for the dilated \(q\)-sample, which induces a corresponding classical fidelity-type quantity after dephasing and marginalisation. There is, to our knowledge, no known general theorem directly linking these two quantities: they are based on different notions of process discrepancy, and the analytic results should therefore not be interpreted as direct bounds on CDR. At the same time, in the examples studied here they exhibit similar qualitative behaviour, which is why CDR remains a useful empirical diagnostic while preserving a clear distinction from the analytically controlled quantity.}

\begin{acknowledgments}
This work was funded by the University of Manchester Dame Kathleen Ollerenshaw Fellowship.
\end{acknowledgments}


\bibliography{apssamp} 

@article{shalizi2001computational,
  title={Computational mechanics: Pattern and prediction, structure and simplicity},
  author={Shalizi, C. R. and Crutchfield, J. P.},
  journal={Journal of Statistical Physics},
  volume={104},
  number={3-4},
  pages={817--879},
  year={2001}
}

@misc{srinivasan2017LearningHQMM,
      title={Learning Hidden Quantum Markov Models}, 
      author={Siddarth Srinivasan and Geoff Gordon and Byron Boots},
      year={2017},
      eprint={1710.09016},
      archivePrefix={arXiv},
      primaryClass={stat.ML},
      url={https://arxiv.org/abs/1710.09016}, 
}

@article{gu2012quantum,
  title={Quantum mechanics can reduce the complexity of classical models},
  author={Gu, M. and Wiesner, K. and Rieper, E. and Vedral, V.},
  journal={Nature Communications},
  volume={3},
  pages={762},
  year={2012}
}

@article{Yang_2018,
   title={Matrix Product States for Quantum Stochastic Modeling},
   volume={121},
   ISSN={1079-7114},
   number={26},
   journal={Physical Review Letters},
   publisher={American Physical Society (APS)},
   author={Yang, Chengran and Binder, Felix C. and Narasimhachar, Varun and Gu, Mile},
   year={2018},
   month=dec }

@article{perezgarcia2007matrixproductstaterepresentations,
  title={Matrix product state representations},
  author={Perez-Garcia, D. and Verstraete, F. and Wolf, M. M. and Cirac, J. I.},
  journal={Quantum Info. Comput.},
  volume={7},
  number={5\&6},
  pages={401--430},
  year={2007}
}

@article{vanderstraeten_tangent-space_2019,
  title={Tangent-space methods for variational optimization of uniform matrix product states},
  author={Vanderstraeten, L. and Haegeman, J. and Verstraete, F.},
  journal={SciPost Physics Lecture Notes},
  pages={007},
  year={2019}
}

@article{yang2024dimensionreductionquantumsampling,
  title={Dimension reduction in quantum sampling of stochastic processes},
  author={Yang, Chengran and Florido-Llin{\`a}s, Marta and Gu, Mile and Elliott, Thomas J},
  journal={npj Quantum Information},
  volume={11},
  number={1},
  pages={34},
  year={2025},
  publisher={Nature Publishing Group UK London}
}

@article{Yang_2020,
   title={Measures of distinguishability between stochastic processes},
   volume={101},
   ISSN={2470-0053},
   number={6},
   journal={Physical Review E},
   publisher={American Physical Society (APS)},
   author={Yang, Chengran and Binder, Felix C. and Gu, Mile and Elliott, Thomas J.},
   year={2020},
   month=jun }

@misc{warden2018speech,
  title={Speech commands: A dataset for limited-vocabulary speech recognition},
  author={Warden, P.},
  year={2018},
  eprint={1804.03209},
  archivePrefix={arXiv},
  primaryClass={cs.CL}
}

@article{Elliott_2021,
   title={Memory compression and thermal efficiency of quantum implementations of nondeterministic hidden Markov models},
   volume={103},
   ISSN={2469-9934},
   number={5},
   journal={Physical Review A},
   publisher={American Physical Society (APS)},
   author={Elliott, Thomas J.},
   year={2021},
   month=may }

@article{Liu_2019,
   title={Optimal stochastic modeling with unitary quantum dynamics},
   volume={99},
   ISSN={2469-9934},
   number={6},
   journal={Physical Review A},
   publisher={American Physical Society (APS)},
   author={Liu, Qing and Elliott, Thomas J. and Binder, Felix C. and Di Franco, Carlo and Gu, Mile},
   year={2019},
   month=jun }

@article{rabiner1989hmm,
  title={A Tutorial on Hidden Markov Models and Selected Applications in Speech Recognition},
  author={Rabiner, Lawrence R.},
  journal={Proceedings of the IEEE},
  year={1989},
}

@misc{stolcke1994bmm,
  title={Best-first Model Merging for Hidden Markov Model Induction},
  author={Stolcke, Andreas and Omohundro, Stephen M.},
  year={1994},
  note={Technical Report TR-94-003, International Computer Science Institute; also arXiv:cmp-lg/9405017},
}

@misc{stolcke1996merging,
  title={Model Merging for Hidden Markov Model Induction},
  author={Stolcke, Andreas and Omohundro, Stephen M.},
  year={1996},
  note={Technical report/manuscript version available online},
}

@article{dupont2005links,
  title={Links between probabilistic automata and hidden Markov models: Probability distributions, learning models and induction algorithms},
  author={Dupont, Pierre and Denis, Fran{\c c}ois and Esposito, Yann},
  journal={Pattern Recognition},
  year={2005},
}

@inproceedings{singh2004psr,
  title={Predictive State Representations: A New Theory for Modeling Dynamical Systems},
  author={Singh, Satinder and James, Michael R. and Rudary, Matthew R.},
  year={2004},
  
}

@inproceedings{dupont2000alergia,
  title={Probabilistic DFA Inference using Kullback-Leibler Divergence and ALERGIA},
  author={Dupont, Pierre},
  booktitle={Proceedings of the International Conference on Machine Learning (ICML)},
  year={2000},
}

@article{Banchi_2024,
   title={Accuracy vs memory advantage in the quantum simulation of stochastic processes},
   volume={5},
   ISSN={2632-2153},
   number={2},
   journal={Machine Learning: Science and Technology},
   publisher={IOP Publishing},
   author={Banchi, Leonardo},
   year={2024},
   month=may, pages={025036} }

@article{Elliott_2020_extreme,
   title={Extreme Dimensionality Reduction with Quantum Modeling},
   volume={125},
   ISSN={1079-7114},
   number={26},
   journal={Physical Review Letters},
   publisher={American Physical Society (APS)},
   author={Elliott, Thomas J. and Yang, Chengran and Binder, Felix C. and Garner, Andrew J. P. and Thompson, Jayne and Gu, Mile},
   year={2020},
   month=dec }

@misc{srinivasan2020quantumtensornetworksstochastic,
      title={Quantum Tensor Networks, Stochastic Processes, and Weighted Automata}, 
      author={Siddarth Srinivasan and Sandesh Adhikary and Jacob Miller and Guillaume Rabusseau and Byron Boots},
      year={2020},
      eprint={2010.10653},
      archivePrefix={arXiv},
      primaryClass={cs.LG},
}

@article{davis1980mfcc,
  title   = {Comparison of Parametric Representations for Monosyllabic Word Recognition in Continuously Spoken Sentences},
  author  = {Davis, Steven B. and Mermelstein, Paul},
  journal = {IEEE Transactions on Acoustics, Speech, and Signal Processing},
  volume  = {28},
  number  = {4},
  pages   = {357--366},
  year    = {1980},
}

@article{dempster1977em,
  title   = {Maximum Likelihood from Incomplete Data via the {EM} Algorithm},
  author  = {Dempster, A. P. and Laird, N. M. and Rubin, D. B.},
  journal = {Journal of the Royal Statistical Society: Series B (Methodological)},
  volume  = {39},
  number  = {1},
  pages   = {1--38},
  year    = {1977},
}

@inproceedings{sculley2010webscale,
  title={Web-Scale {K}-Means Clustering},
  author={Sculley, David},
  booktitle={Proceedings of the 19th International Conference on World Wide Web (WWW)},
  pages={1177--1178},
  year={2010},
}

@article{gammelmark2014hidden,
  title={Hidden Markov model of atomic quantum jump dynamics in an optically probed cavity},
  author={Gammelmark, S and M{\o}lmer, K and Alt, W and Kampschulte, T and Meschede, D},
  journal={Physical Review A},
  volume={89},
  number={4},
  pages={043839},
  year={2014},
  publisher={APS}}

@inproceedings{karlof2003hidden,
  title={Hidden Markov model cryptanalysis},
  author={Karlof, Chris and Wagner, David},
  booktitle={International Workshop on Cryptographic Hardware and Embedded Systems},
  pages={17--34},
  year={2003},
  organization={Springer}}

@book{bhar2004hidden,
  title={Hidden Markov models: applications to financial economics},
  author={Bhar, Ramaprasad and Hamori, Shigeyuki},
  volume={40},
  year={2004},
  publisher={Springer Science \& Business Media}}

@inproceedings{seymore1999learning,
  title={Learning hidden Markov model structure for information extraction},
  author={Seymore, Kristie and McCallum, Andrew and Rosenfeld, Roni},
  booktitle={AAAI-99 workshop on machine learning for information extraction},
  pages={37--42},
  year={1999}}

@article{fine1998hierarchical,
  title={The hierarchical hidden Markov model: Analysis and applications},
  author={Fine, Shai and Singer, Yoram and Tishby, Naftali},
  journal={Machine learning},
  volume={32},
  number={1},
  pages={41--62},
  year={1998},
  publisher={Springer}}

@inproceedings{ghahramani1996factorial,
  title={Factorial hidden Markov models},
  author={Ghahramani, Zoubin and Jordan, Michael I},
  booktitle={Advances in Neural Information Processing Systems},
  pages={472--478},
  year={1996}}

@article{stanke2003gene,
  title={Gene prediction with a hidden Markov model and a new intron submodel},
  author={Stanke, Mario and Waack, Stephan},
  journal={Bioinformatics},
  volume={19},
  number={suppl\_2},
  pages={ii215--ii225},
  year={2003},
  publisher={Oxford University Press}}

@article{baldi1994hidden,
  title={Hidden Markov models of biological primary sequence information.},
  author={Baldi, Pierre and Chauvin, Yves and Hunkapiller, Tim and McClure, Marcella A},
  journal={Proceedings of the National Academy of Sciences},
  volume={91},
  number={3},
  pages={1059--1063},
  year={1994},
  publisher={National Acad Sciences}}

@article{krogh2001predicting,
  title={Predicting transmembrane protein topology with a hidden Markov model: application to complete genomes},
  author={Krogh, Anders and Larsson, Bjo{\`E}rn and Von Heijne, Gunnar and Sonnhammer, Erik LL},
  journal={Journal of Molecular Biology},
  volume={305},
  number={3},
  pages={567--580},
  year={2001},
  publisher={Elsevier}}

@article{rabiner1986introduction,
  title={An introduction to hidden Markov models},
  author={Rabiner, Lawrence and Juang, B},
  journal={IEEE Acoustics, Speech and Signal Processing magazine},
  volume={3},
  number={1},
  pages={4--16},
  year={1986},
  publisher={IEEE}}

@phdthesis{upper1997theory,
  title={Theory and algorithms for hidden Markov models and generalized hidden Markov models},
  author={Upper, Daniel Ray},
  year={1997},
  school={University of California, Berkeley}}

@article{ruebeck2018prediction,
  title={Prediction and generation of binary {M}arkov processes: {C}an a finite-state fox catch a {M}arkov mouse?},
  author={Ruebeck, Joshua B and James, Ryan G and Mahoney, John R and Crutchfield, James P},
  journal={Chaos: An Interdisciplinary Journal of Nonlinear Science},
  volume={28},
  number={1},
  pages={013109},
  year={2018},
  publisher={AIP Publishing LLC}}

@article{ghafari2019interfering,
  title={Interfering trajectories in experimental quantum-enhanced stochastic simulation},
  author={Ghafari, Farzad and Tischler, Nora and Di Franco, Carlo and Thompson, Jayne and Gu, Mile and Pryde, Geoff J},
  journal={Nature Communications},
  volume={10},
  number={1},
  pages={1630},
  year={2019},
  publisher={Nature Publishing Group}}

@article{ghafari2019dimensional,
  title={Dimensional Quantum Memory Advantage in the Simulation of Stochastic Processes},
  author={Ghafari, Farzad and Tischler, Nora and Thompson, Jayne and Gu, Mile and Shalm, Lynden K and Verma, Varun B and Nam, Sae Woo and Patel, Raj B and Wiseman, Howard M and Pryde, Geoff J},
  journal={Physical Review X},
  volume={9},
  number={4},
  pages={041013},
  year={2019},
  publisher={APS}}

@article{crutchfield1997statistical,
  title={Statistical complexity of simple one-dimensional spin systems},
  author={Crutchfield, James P and Feldman, David P},
  journal={Physical Review E},
  volume={55},
  number={2},
  pages={R1239},
  year={1997},
  publisher={APS}}

@article{mahoney2016occam,
  title={Occam's quantum strop: {S}ynchronizing and compressing classical cryptic processes via a quantum channel},
  author={Mahoney, John R and Aghamohammadi, Cina and Crutchfield, James P},
  journal={Scientific Reports},
  volume={6},
  pages={20495},
  year={2016},
  publisher={Nature Publishing Group}}

@article{palsson2017experimentally,
  title={Experimentally modeling stochastic processes with less memory by the use of a quantum processor},
  author={Palsson, Matthew S and Gu, Mile and Ho, Joseph and Wiseman, Howard M and Pryde, Geoff J},
  journal={Science Advances},
  volume={3},
  number={2},
  pages={e1601302},
  year={2017},
  publisher={American Association for the Advancement of Science}}

@article{binder2018practical,
  title={Practical Unitary Simulator for Non-{M}arkovian Complex Processes},
  author={Binder, Felix C and Thompson, Jayne and Gu, Mile},
  journal={Physical Review Letters},
  volume={120},
  number={24},
  pages={240502},
  year={2018},
  publisher={APS}}

@article{garner2017provably,
  title={Provably unbounded memory advantage in stochastic simulation using quantum mechanics},
  author={Garner, Andrew J P and Liu, Qing and Thompson, Jayne and Vedral, Vlatko and others},
  journal={New Journal of Physics},
  volume={19},
  number={10},
  pages={103009},
  year={2017},
  publisher={IOP Publishing}}

@article{aghamohammadi2017extreme,
  author    = {Aghamohammadi, Cina and Mahoney, John R and Crutchfield, James P},
  title     = {Extreme Quantum Advantage when Simulating Classical Systems with Long-Range Interaction},
  journal   = {Scientific Reports},
  year      = {2017},
  volume    = {7},
  publisher = {Nature Publishing Group}}

@article{elliott2018superior,
  title={Superior memory efficiency of quantum devices for the simulation of continuous-time stochastic processes},
  author={Elliott, Thomas J and Gu, Mile},
  journal={npj Quantum Information},
  volume={4},
  pages={18},
  year={2018},
  publisher={Nature Publishing Group}}

@article{elliott2019memory,
  author={Elliott, Thomas J and Garner, Andrew J P and Gu, Mile},
  title = "{Memory-efficient tracking of complex temporal and symbolic dynamics with quantum simulators}",
  journal={New Journal of Physics},
  year = {2019},
  volume={21},
  pages={013021}}

@article{aghamohammadi2018extreme,
  title={Extreme Quantum Memory Advantage for Rare-Event Sampling},
  author={Aghamohammadi, Cina and Loomis, Samuel P and Mahoney, John R and Crutchfield, James P},
  journal={Physical Review X},
  volume={8},
  number={1},
  pages={011025},
  year={2018},
  publisher={APS}}

@article{riechers2016minimized,
  title={Minimized state complexity of quantum-encoded cryptic processes},
  author={Riechers, Paul M and Mahoney, John R and Aghamohammadi, Cina and Crutchfield, James P},
  journal={Physical Review A},
  volume={93},
  number={5},
  pages={052317},
  year={2016},
  publisher={APS}}

@article{khintchine1934korrelationstheorie,
  title={Korrelationstheorie der station{\"a}ren stochastischen {P}rozesse},
  author={Khintchine, Alexander},
  journal={Mathematische Annalen},
  volume={109},
  number={1},
  pages={604--615},
  year={1934},
  publisher={Springer}}

@article{marzen2015informational,
  title={Informational and causal architecture of discrete-time renewal processes},
  author={Marzen, Sarah E and Crutchfield, James P},
  journal={Entropy},
  volume={17},
  number={7},
  pages={4891--4917},
  year={2015},
  publisher={Multidisciplinary Digital Publishing Institute}}

@article{crutchfield2012between,
  title={Between order and chaos},
  author={Crutchfield, James P},
  journal={Nature Physics},
  volume={8},
  number={1},
  pages={17--24},
  year={2012},
  publisher={Nature Publishing Group}}

@article{wu2023implementing,
  title={Implementing quantum dimensionality reduction for non-Markovian stochastic simulation},
  author={Wu, Kang-Da and Yang, Chengran and He, Ren-Dong and Gu, Mile and Xiang, Guo-Yong and Li, Chuan-Feng and Guo, Guang-Can and Elliott, Thomas J},
  journal={Nature Communications},
  volume={14},
  number={1},
  pages={2624},
  year={2023},
  publisher={Nature Publishing Group UK London}
}

@article{elliott2024embedding,
  title={Embedding memory-efficient stochastic simulators as quantum trajectories},
  author={Elliott, Thomas J and Gu, Mile},
  journal={Physical Review A},
  volume={109},
  number={2},
  pages={022434},
  year={2024},
  publisher={APS}
}

@article{han2018unsupervised,
  title={Unsupervised generative modeling using matrix product states},
  author={Han, Zhao-Yu and Wang, Jun and Fan, Heng and Wang, Lei and Zhang, Pan},
  journal={Physical Review X},
  volume={8},
  number={3},
  pages={031012},
  year={2018},
  publisher={APS}
}

@inproceedings{adhikary2021quantum,
  title={Quantum tensor networks, stochastic processes, and weighted automata},
  author={Adhikary, Sandesh and Srinivasan, Siddarth and Miller, Jacob and Rabusseau, Guillaume and Boots, Byron},
  booktitle={International Conference on Artificial Intelligence and Statistics},
  pages={2080--2088},
  year={2021},
  organization={PMLR}
}

@article{glasser2019expressive,
  title={Expressive power of tensor-network factorizations for probabilistic modeling},
  author={Glasser, Ivan and Sweke, Ryan and Pancotti, Nicola and Eisert, Jens and Cirac, Ignacio},
  journal={Advances in neural information processing systems},
  volume={32},
  year={2019}}

@article{Stoudenmire2017,
archivePrefix = {arXiv},
arxivId = {1801.00315},
author = {Stoudenmire, E. M.},
eprint = {1801.00315},
journal = {arXiv preprint},
title = {{Learning Relevant Features of Data with Multi-scale Tensor Networks}},
year = {2017}
}

@article{schuld2018quantum,
  title={Quantum Advantages},
  author={Schuld, Maria and Petruccione, Francesco and Schuld, Maria and Petruccione, Francesco},
  journal={Supervised Learning with Quantum Computers},
  pages={127--137},
  year={2018},
  publisher={Springer}
}

@article{elliott2022quantum,
  title={Quantum adaptive agents with efficient long-term memories},
  author={Elliott, Thomas J and Gu, Mile and Garner, Andrew JP and Thompson, Jayne},
  journal={Physical Review X},
  volume={12},
  pages={011007},
  year={2022},
  publisher={APS}
}

@article{zonnios2025quantum,
  title={Quantum generation of stochastic processes: spectral invariants and memory bounds},
  author={Zonnios, Magdalini and Boyd, Alexander and Binder, Felix},
  journal={New Journal of Physics},
  year={2025}
}

@article{monras2010hidden,
  title={Hidden quantum {M}arkov models and non-adaptive read-out of many-body states},
  author={Monras, Alex and Beige, Almut and Wiesner, Karoline},
  journal={arXiv preprint arXiv:1002.2337},
  year={2010}
}

@article{Monras16_Quantum,
	author = {Monràs, Alex and Winter, Andreas},
	title = {Quantum learning of classical stochastic processes: {T}he completely positive realization problem},
	journal = {Journal of Mathematical Physics},
	volume = {57},
	number = {1},
	pages = {015219},
	year = {2016},
	month = {01},
	issn = {0022-2488},
}

@article{riechers2025identifiability,
  title={Identifiability and minimality bounds of quantum and post-quantum models of classical stochastic processes},
  author={Riechers, Paul M and Elliott, Thomas J},
  journal={arXiv preprint arXiv:2509.03004},
  year={2025}
}

@article{Fanizza24_Quantum,
	title={Quantum theory in finite dimension cannot explain every general process with finite memory},
	author={Fanizza, Marco and Lumbreras, Josep and Winter, Andreas},
	journal={Communications in Mathematical Physics},
	volume={405},
	number={2},
	pages={50},
	year={2024},
	publisher={Springer}
}

@article{crutchfield1989inferring,
  title={Inferring statistical complexity},
  author={Crutchfield, James P and Young, Karl},
  journal={Physical Review Letters},
  volume={63},
  number={2},
  pages={105},
  year={1989},
  publisher={APS}}

@article{crutchfield1994calculi,
  title={The calculi of emergence: computation, dynamics and induction},
  author={Crutchfield, James P},
  journal={Physica D: Nonlinear Phenomena},
  volume={75},
  number={1-3},
  pages={11--54},
  year={1994},
  publisher={Elsevier}}

@inproceedings{lohr2009non,
  title={Non-sufficient memories that are sufficient for prediction},
  author={L{\"o}hr, Wolfgang and Ay, Nihat},
  booktitle={International Conference on Complex Sciences},
  pages={265--276},
  year={2009},
  organization={Springer}}

@article{elliott2021quantum,
   author = {Elliott, T J},
    title = "{Quantum coarse graining for extreme dimension reduction in modeling stochastic temporal dynamics}",
year={2021},
journal={PRX Quantum},
volume={2},
pages={020342}}

\appendix

\section{Properties of the dilation}
In this Appendix we collect the technical arguments underlying the structural
properties of the dilation We work throughout with the dilated
$q$–sample iMPS for the process over $(x,y)$ and its compressed version at
bond dimension $\tilde d$, as constructed in
Secs.~\ref{subsec:qsample} and~\ref{subsec:truncation}.

We start by recording the basic structural properties of the dilation
defined in Sec.~\ref{subsec:dilation}. Recall that the original HMM
$(\mathcal S,\mathcal X,\{T^{x}\})$ is mapped to a dilated HMM
$(\mathcal S,\mathcal X\times\mathcal Y,\{T^{(x,y)}\})$ with transition
tensor
\begin{equation}
T^{(x,y)}_{s's} =
\begin{cases}
T^{x}_{s's}, & \text{if } T^{x}_{s's} > 0 \text{ and } y = f(s,s',x),\\[4pt]
0, & \text{otherwise,}
\end{cases}
\end{equation}
where the labelling function
$f:\mathcal S\times\mathcal S\times\mathcal X\to\mathcal Y$ is assumed to be
injective in its second argument for each fixed $(s,x)$.

\begin{lemma}[Deterministicity of the dilation]
For each state $s\in\mathcal S$ and composite symbol
$(x,y)\in\mathcal X\times\mathcal Y$ there is at most one successor
$s'\in\mathcal S$ with $T^{(x,y)}_{s's} > 0$.
\end{lemma}

\begin{proof}
Fix $s$ and $x$. For each $s'$ with $T^{x}_{s's}>0$ the construction assigns
a single label $y=f(s,s',x)$. The injectivity of $f$ for fixed $(s,x)$ implies
that $s'\neq s''$ gives $f(s,s',x)\neq f(s,s'',x)$. Hence for any pair
$(x,y)$ there is at most one $s'$ such that $T^{(x,y)}_{s's}>0$.
\end{proof}

\begin{lemma}[Preservation of observable statistics]
Let $\mathcal P$ be the process on $\mathcal X$ generated by the original
HMM and let $\mathcal P_{\mathrm{dil}}$ be the process on $\mathcal X$ obtained
from the dilated HMM by marginalising over $\mathcal Y$. Then for all block
lengths $L$ and strings $(x_1,\dots,x_L)$,
\begin{equation}
\Pr_{\mathcal P}(x_1,\dots,x_L)
=
\Pr_{\mathcal P_{\mathrm{dil}}}(x_1,\dots,x_L).
\end{equation}
\end{lemma}

\begin{proof}
For any $s,s'$ and $x$ we have
\begin{equation}
\sum_{y\in\mathcal Y} T^{(x,y)}_{s's}
= \sum_{y: y=f(s,s',x)} T^{x}_{s's}
= T^{x}_{s's},
\end{equation}
since at most one $y$ contributes. Thus the conditional distribution of
$(S_{t+1},X_t)$ given $S_t$ is unchanged when we marginalise
over $Y_t$. Iterating this equality along the chain yields equality of all
finite-length block probabilities on $\mathcal X$.
\end{proof}

\begin{lemma}[Ergodicity of the dilation]
If the original HMM is stationary and ergodic, then the dilated HMM is also
stationary and ergodic.
\end{lemma}

\begin{proof}
The underlying Markov chain on $\mathcal S$, obtained by summing over
outputs, is the same in both models:
\begin{equation}
\sum_{x,y} T^{(x,y)}_{s's}
= \sum_x T^{x}_{s's}.
\end{equation}
Irreducibility and aperiodicity of this chain are therefore inherited from
the original HMM, and the stationary distribution over $\mathcal S$ carries
over unchanged. Hence the dilated HMM is stationary and ergodic.
\end{proof}

\begin{lemma}[Square-root tensors reproduce word statistics for unifilar generators]
\label{lem:sqrtT_correct_statistics}
Let $(\mathcal S,\mathcal A,\{T^{a}\}_{a\in\mathcal A})$ be a finite-state,
edge-emitting \emph{unifilar} HMM for a stationary process on alphabet
$\mathcal A$ (in our application $\mathcal A=\mathcal X\times\mathcal Y$).
Write the transition probabilities as
$T^{a}_{s's}=\Pr(S_{t+1}=s',A_t=a\mid S_t=s)$.
Assume the underlying hidden-state Markov chain
$P:=\sum_{a\in\mathcal A} T^{a}$ is ergodic with stationary distribution
$\pi$.

Define site matrices $A^{a}\in\mathbb R^{|\mathcal S|\times|\mathcal S|}$ by
\begin{equation}
A^{a}_{s's} := \sqrt{T^{a}_{s's}}.
\end{equation}
Let $\rho_\star := \mathrm{diag}(\pi)$ be the diagonal matrix with entries
$\pi_s$.
Then for every word $w=a_1 a_2 \cdots a_L\in\mathcal A^L$,
\begin{equation}
\Pr_{\mathrm{HMM}}(w)
=
\mathrm{Tr}\!\left(
A^{a_L}\cdots A^{a_1}\,\rho_\star\,
A^{a_1\dagger}\cdots A^{a_L\dagger}
\right).
\label{eq:wordprob_sqrtT}
\end{equation}
Equivalently, the completely positive maps
$\mathcal E_a(\rho):=A^{a}\rho A^{a\dagger}$
define a quantum instrument whose word probabilities from the stationary bond
state $\rho_\star$ agree exactly with the original unifilar HMM.
\end{lemma}

\begin{proof}
Fix a word $w=a_1\cdots a_L$ and write $M_w := A^{a_L}\cdots A^{a_1}$.
Since $\rho_\star$ is diagonal,
\begin{align}
\mathrm{Tr}(M_w \rho_\star M_w^\dagger)
&= \sum_{s_0\in\mathcal S} (\rho_\star)_{s_0 s_0}\,
\sum_{s_L\in\mathcal S} |(M_w)_{s_L s_0}|^2 \nonumber\\
&= \sum_{s_0\in\mathcal S} \pi_{s_0}\,
\sum_{s_L\in\mathcal S} |(M_w)_{s_L s_0}|^2.
\label{eq:trace_expand}
\end{align}

Now use unifilarity: for each current state $s$ and symbol $a$ there is at most
one successor state $\delta(s,a)$ with $T^{a}_{\delta(s,a),s}>0$.
Hence, for fixed $s_0$, there is at most one compatible internal-state path
$s_1,s_2,\dots,s_L$ satisfying $s_t=\delta(s_{t-1},a_t)$ for all $t$.
Therefore each column $s_0$ of $M_w$ has at most one nonzero entry, located at
row $s_L=\delta(\cdots\delta(s_0,a_1)\cdots,a_L)$ when the path exists.
Moreover, along that path,
\begin{equation}
(M_w)_{s_L s_0}
= \prod_{t=1}^L \sqrt{T^{a_t}_{s_t s_{t-1}}}
\quad \Longrightarrow \quad
|(M_w)_{s_L s_0}|^2
= \prod_{t=1}^L T^{a_t}_{s_t s_{t-1}}.
\end{equation}
Summing over $s_L$ in \eqref{eq:trace_expand} therefore removes the final-state
index without introducing cross-terms:
\begin{equation}
\sum_{s_L} |(M_w)_{s_L s_0}|^2
= \prod_{t=1}^L T^{a_t}_{s_t s_{t-1}},
\end{equation}
with the understanding that the product is zero if the path does not exist.
Substituting back into \eqref{eq:trace_expand} gives
\begin{equation}
\mathrm{Tr}(M_w \rho_\star M_w^\dagger)
= \sum_{s_0\in\mathcal S} \pi_{s_0}\,
\prod_{t=1}^L T^{a_t}_{s_t s_{t-1}}.
\end{equation}
The right-hand side is exactly the standard unifilar-HMM expression for the
stationary word probability $\Pr_{\mathrm{HMM}}(w)$ (sum over initial state
weighted by $\pi$, with the internal path fixed by unifilarity). This proves
\eqref{eq:wordprob_sqrtT}.
\end{proof}

\emph{Relation to the $q$-sample construction}
Lemma~\ref{lem:sqrtT_correct_statistics} shows that the square-root tensors
always reproduce the \emph{classical} word statistics for any unifilar
presentation when probabilities are extracted via the induced instrument.
For predictive presentations (e.g.\ $\varepsilon$-machines) this iMPS further
coincides with the $q$-sample state in the sense of Ref.~\cite{Yang_2018};
for general nonpredictive unifilar presentations one should not assume this
stronger identification, even though the measured statistics agree.

\subsection{\rev{Implemented labelling heuristics}}
\label{app:labelling_heuristics}

{

The dilation of Sec.~\ref{subsec:dilation} requires choosing an admissible
labelling function $f$ satisfying the injectivity condition
Eq.~\eqref{eq:f_injective}. All admissible choices produce a deterministic
HMM that exactly reproduces the original process before truncation; they differ
only in how auxiliary labels $y\in\mathcal Y$ are assigned to the nonzero
outgoing transitions from each $(s,x)$ pair, and hence in the tensor structure
seen by the subsequent variational truncation.

In our numerical study (Sec.~\ref{subsec:labelling}) we implement the following
four heuristic strategies. In each case, for a fixed source state $s$ and
emitted symbol $x$, let $s'_1,\dots,s'_{k}$ denote the destination states with
$T^x_{s'_i s}>0$, and let $|\mathcal Y|=d_y$ be the minimal auxiliary
alphabet size.

\begin{enumerate}[label=(\roman*)]
  \item \emph{Sequential.}
  Sort the destination states by state index in ascending order and assign
  auxiliary labels $y=0,1,\dots,k-1$ in that order. This is a deterministic
  baseline that depends only on the state-labelling convention.

  \item \emph{Probability descending.}
  Sort the destination states by their transition probabilities
  $T^x_{s'_i s}$ in decreasing order and assign $y=0$ to the highest-probability
  transition, $y=1$ to the next, and so on.

  \item \emph{Probability ascending.}
  As above, but with the destination states sorted in increasing order of
  transition probability, so that $y=0$ is assigned to the least probable
  transition.

  \item \emph{Random permutation.}
  Assign auxiliary labels according to a pseudorandom permutation of the
  destination states. When a base seed is supplied, the assignment is
  reproducible.
\end{enumerate}

All four strategies use the same minimal auxiliary alphabet size $d_y$ and
produce dilated tensors satisfying $\sum_y T^{(x,y)}_{s's}=T^x_{s's}$ for all
$s,s',x$. They are therefore all exact prior to truncation. As illustrated
in Fig.~\ref{fig:labelling_cdr}, the choice of strategy can materially affect
the Schmidt spectrum of the dilated iMPS and hence the compression performance.
}

\section{Error Bounds}
\label{app:bounds}
\subsection{From quantum fidelity rate to a classical Bhattacharyya rate}
\label{subsec:qfdr_to_rb_app}

Let $\ket{P_{xy}}$ and $\ket{\tilde P_{xy}}$ denote the infinite-chain
$q$--sample states of the exact and truncated dilated processes over
$\mathcal X\times\mathcal Y$. For concreteness we define the quantum fidelity
divergence rate (QFDR) as
\begin{equation}
  R_F(\ket{P_{xy}},\ket{\tilde P_{xy}})
  = -\lim_{L\to\infty} \frac{1}{2L}
    \log_2 F\bigl(\rho_L,\tilde\rho_L\bigr),
\end{equation}
where $\rho_L$ and $\tilde\rho_L$ are the reduced density operators on
$L$ sites and $F(\rho,\sigma)$ is the quantum fidelity.

The following lemma relates QFDR to a classical fidelity divergence rate on $\mathcal X$, defined by the blockwise
Bhattacharyya coefficient
$\sum_{\vec x}\sqrt{P^{(L)}(\vec x)\,\tilde P^{(L)}(\vec x)}$.
\begin{lemma}[Classical fidelity rate bounded by QFDR]
\label{lem:cfdr_qfdr_app}
Let $\mathcal P$ and $\tilde{\mathcal P}$ be the processes on $\mathcal X$
obtained by marginalising the exact and truncated dilated $q$--samples over
$\mathcal Y$. Define the classical fidelity divergence rate
\begin{equation}
  R_F(\mathcal P,\tilde{\mathcal P})
  := -\lim_{L\to\infty}\frac{1}{2L}\log_2
  \sum_{\vec x\in\mathcal X^L}\sqrt{P^{(L)}(\vec x)\,\tilde P^{(L)}(\vec x)} .
\end{equation}
Then
\begin{equation}
  R_F(\mathcal P,\tilde{\mathcal P})
  \;\le\;
  R_F(\ket{P_{xy}},\ket{\tilde P_{xy}}).
\end{equation}
\end{lemma}

\begin{proof}
Consider the reduced density operators $\rho_L$ and $\tilde\rho_L$ on $L$
sites. Apply two completely positive trace-preserving (CPTP) maps to both
states: (i) dephasing in the computational basis
$\{\ket{x_1,y_1;\dots;x_L,y_L}\}$; and (ii) partial trace over the auxiliary
outputs $y_1,\dots,y_L$. Let $\Phi$ denote the composition of these maps.
Fidelity is monotone under CPTP maps, so
\begin{equation}
  F(\rho_L,\tilde\rho_L)
  \;\le\;
  F\bigl(\Phi(\rho_L),\Phi(\tilde\rho_L)\bigr).
\end{equation}
The outputs $\Phi(\rho_L)$ and $\Phi(\tilde\rho_L)$ are diagonal density
matrices on $\mathcal X^L$ encoding the classical block distributions
$P^{(L)}$ and $\tilde P^{(L)}$. Their quantum fidelity reduces to the
classical fidelity coefficient
\(\sum_{\vec x}\sqrt{P^{(L)}(\vec x)\tilde P^{(L)}(\vec x)}\).
Taking $-\frac{1}{2L}\log_2(\cdot)$ and passing to $L\to\infty$ yields the
claim.
\end{proof}

\subsection{Truncation and quantum fidelity rate}

Let $\{\lambda_k\}_{k=1}^{d_s}$ be the Schmidt coefficients of the dilated
iMPS across a fixed bond, ordered so that
$\lambda_1 \ge \lambda_2 \ge \cdots \ge \lambda_{d_s}$ and
$\sum_k \lambda_k = 1$. For a target bond dimension $\tilde d$ we define the
discarded tail weight
\begin{equation}
  \varepsilon_{\tilde d}
  := \sum_{k>\tilde d} \lambda_k .
\end{equation}

The effect of discarding Schmidt weight at a single bond of a normal iMPS and
replacing it with a variationally optimal truncated iMPS is analysed in
Ref.~\cite{yang2024dimensionreductionquantumsampling}. We state the relevant
scaling result:

\begin{lemma}[QFDR versus Schmidt tail]
\label{lem:qfdr_tail_app}
For the normal dilated iMPS and its variational truncation at bond dimension
$\tilde d$, there exists a constant $c>0$ (depending on the local physical
dimension and the spectral gap of the transfer operator) such that, for
sufficiently small $\varepsilon_{\tilde d}$,
\begin{equation}
  R_F(\ket{P_{xy}},\ket{\tilde P_{xy}})
  \;\le\;
  c\,\varepsilon_{\tilde d}.
\end{equation}
\end{lemma}

\begin{proof}[Sketch]
Discarding Schmidt weight $\varepsilon_{\tilde d}$ and constructing the
optimal variational approximation at bond dimension $\tilde d$ perturbs the
iMPS by an amount linear in $\varepsilon_{\tilde d}$ in appropriate norms.
For a normal iMPS with a primitive transfer operator, a finite-block overlap
bound can be derived in terms of $\varepsilon_{\tilde d}$ and the spectral
gap. Passing to the infinite-chain limit yields the stated bound on the
fidelity divergence rate. The full argument is given in
Ref.~\cite{yang2024dimensionreductionquantumsampling}.
\end{proof}

Combining Lemmas~\ref{lem:cfdr_qfdr_app} and~\ref{lem:qfdr_tail_app} yields
\begin{equation}
  R_F(\mathcal P,\tilde{\mathcal P})
  \;\le\;
  c\,\varepsilon_{\tilde d}.
  \label{eq:cfdr_tail_app}
\end{equation}

\subsection{Tail weight and Schmidt entropy}

We now bound the tail weight $\varepsilon_{\tilde d}$ in terms of the Schmidt
entropy
\begin{equation}
  H(\lambda)
  := -\sum_k \lambda_k \log_2 \lambda_k .
\end{equation}

\begin{lemma}[Entropy controls tail]
\label{lem:entropy_tail_app}
Let $\{\lambda_k\}$ be a probability distribution in non-increasing order
and let $\varepsilon_{\tilde d} = \sum_{k>\tilde d}\lambda_k$ for some
integer $\tilde d\ge 2$. Then
\begin{equation}
  \varepsilon_{\tilde d}
  \;\le\;
  \frac{H(\lambda)}{\log_2 \tilde d}.
\end{equation}
\end{lemma}

\begin{proof}
For $k>\tilde d$ we have $\lambda_k \le \lambda_{\tilde d}$. Since
$\sum_{k\le\tilde d}\lambda_k \le 1$, it follows that
$\tilde d\,\lambda_{\tilde d} \le 1$ and hence
$\lambda_{\tilde d} \le 1/\tilde d$. Thus
$\lambda_k \le 1/\tilde d$ for all $k>\tilde d$.

Restricting the entropy sum to the tail gives
\begin{align}
  H(\lambda)
  &= -\sum_k \lambda_k \log_2 \lambda_k
   \ge -\sum_{k>\tilde d} \lambda_k \log_2 \lambda_k \nonumber\\
  &= \sum_{k>\tilde d} \lambda_k \log_2 \frac{1}{\lambda_k}.
\end{align}
For $k>\tilde d$ we have $\lambda_k \le 1/\tilde d$, so
$\log_2 (1/\lambda_k) \ge \log_2 \tilde d$, and therefore
\begin{equation}
  H(\lambda)
  \ge \sum_{k>\tilde d} \lambda_k \log_2 \tilde d
  = \varepsilon_{\tilde d} \log_2 \tilde d.
\end{equation}
Rearranging yields the claimed inequality.
\end{proof}

Substituting Lemma~\ref{lem:entropy_tail_app} into Eq.~\eqref{eq:cfdr_tail_app}
gives the entropy-based bound quoted in the main text:
\begin{equation}
  R_F(\mathcal P,\tilde{\mathcal P})
  \;\le\;
  c\,\frac{H(\lambda)}{\log_2 \tilde d}.
  \label{eq:cfdr_entropy_app}
\end{equation}

\subsection{Slice matrix and support of the stationary bond state}

To relate $H(\lambda)$ to an algebraic quantity depending on the dilation, we
consider the stationary bond density operator $\rho_\star$ of the dilated
iMPS. In canonical form, $\rho_\star$ is the fixed point of the bond channel
\begin{equation}
  \mathcal E(\rho)
  = \sum_{x\in\mathcal X}\sum_{y\in\mathcal Y}
    A^{(x,y)} \rho A^{(x,y)\dagger}.
\end{equation}
The eigenvalues of $\rho_\star$ are the Schmidt coefficients $\{\lambda_k\}$
across the corresponding bond.

Define the slice matrix $K$ by horizontally concatenating all nonzero site
tensors $A^{(x,y)}$, viewed as linear maps on the bond space:
\begin{equation}
  K =
  \bigl[\,A^{(x_1,y_1)}\;\; A^{(x_2,y_2)}\;\; \cdots \bigr].
\end{equation}
The column space of $K$ is the span of the ranges of the individual
$A^{(x,y)}$.

\begin{lemma}[Support contained in the slice span]
\label{lem:support_K_app}
Let $\rho_\star$ be the stationary bond state of the dilated iMPS. Then the
support of $\rho_\star$ is contained in the column space of $K$, and hence
\begin{equation}
  \operatorname{rank}(\rho_\star)
  \le \operatorname{rank}(K).
\end{equation}
\end{lemma}

\begin{proof}
Let $\rho_0$ be any initial positive operator with full support. Iterating
the bond channel gives $\rho_{n+1} = \mathcal E(\rho_n)$. Each term
$A^{(x,y)} \rho_n A^{(x,y)\dagger}$ has columns in the range of $A^{(x,y)}$,
and therefore in the column space of $K$. Hence $\rho_n$ has support
contained in that space for all $n\ge 1$. For a primitive channel
$\mathcal E$, the sequence $\rho_n$ converges to the unique fixed point
$\rho_\star$ as $n\to\infty$. The column space of $K$ is closed, so
$\rho_\star$ also has support contained in it. The rank bound follows.
\end{proof}

The Schmidt entropy is the von Neumann entropy of $\rho_\star$ in bits. For
any density matrix of rank $r$, the entropy is maximised by the uniform
distribution on its support and satisfies $H(\lambda) \le \log_2 r$. Applying
this with $r=\operatorname{rank}(\rho_\star)$ and using
Lemma~\ref{lem:support_K_app} yields
\begin{equation}
  H(\lambda)
  \le \log_2 \operatorname{rank}(\rho_\star)
  \le \log_2 \operatorname{rank}(K).
  \label{eq:entropy_rankK_app}
\end{equation}

\subsection{Combined slice-rank bound}

Combining Eq.~\eqref{eq:cfdr_entropy_app} with
Eq.~\eqref{eq:entropy_rankK_app} gives the slice-rank bound
\begin{equation}
  R_F(\mathcal P,\tilde{\mathcal P})
  \;\le\;
  c\,\frac{\log_2 \operatorname{rank} K}{\log_2 \tilde d},
\end{equation}
which appears as Eq.~\eqref{eq:cfdr_rank_bound} in the main text. The rank of
$K$ depends explicitly on the chosen labelling function $f$ through the set
of nonzero slices $A^{(x,y)}$, making this a genuinely label-aware
certificate that links the achievable compression error to the structure
induced by the dilation.

\end{document}